\documentclass[10pt,journal]{IEEEtran}

\usepackage{graphicx,cite}
\usepackage{multirow,amsmath, mathrsfs,soul,ulem}
\usepackage{algorithm,algorithmic,rotating,url}
\usepackage{fancybox,subfigure,amsthm,afterpage}
\usepackage{amsfonts}
\usepackage{times,amsmath,bm,amssymb,color,setspace,multirow,stfloats,extpfeil}
\definecolor{orange}{cmyk}{0,0.5,1,0}
\definecolor{green}{cmyk}{1,0.4,.8,0}
\definecolor{blue}{rgb}{0.2,0.3,0.8}
\definecolor{red}{rgb}{0.8,0.1,0.1}

\newtheorem{lem}{Lemma}
\newtheorem{cor}{Corollary}

\newtheorem{thm}{Theorem}
\newcommand{\subparagraph}{}
\usepackage[compact]{titlesec}
\usepackage{bbm}
\titlespacing{\section}{0pt}{*0}{*0}
\titlespacing{\subsection}{0pt}{*0}{*0}
\titlespacing{\subsubsection}{0pt}{*0}{*0}
\usepackage{multicol}

\title{On Bandwidth Constrained Distributed Detection of a Deterministic Signal in Correlated Noise \thanks{This work is supported by the National Science Foundation under grants CCF-1341966 and CCF-1319770.}}

\author{\normalsize Nahal Maleki, ~\IEEEmembership{Member,~IEEE,}  Azadeh Vosoughi,~\IEEEmembership{Senior Member,~IEEE} 
}

\begin{document}
	
\maketitle
\vspace{-8cm}

\begin{abstract} 
We consider a Neyman-Pearson (NP) distributed binary detection problem in a bandwidth constrained wireless sensor network, where the fusion center (FC) is responsible for fusing signals received from sensors and making a final decision about the presence or absence of a signal source in correlated Gaussian noises. Given this signal model, our goals are (i) to investigate whether or not randomized transmission can improve detection performance, under communication rate constraint, and (ii) to explore 
how the correlation among observation noises would impact performance. To achieve these goals, we propose two novel schemes that combine the concepts of censoring and randomized transmission (which we name CRT-I and CRT-II schemes) and compare them with pure censoring scheme. In CRT (pure censoring) schemes we map randomly (deterministically) a sensor's observation to a ternary transmit symbol $u_k \in \{-1,0,1\}$ where ``$0$'' corresponds to no transmission (sensor censors).
Assuming sensors transmit $u_k$'s over orthogonal fading channels, we formulate and address two system-level constrained optimization problems: in the first problem we minimize the probability of miss detection at the FC, subject to constraints on the probabilities of transmission and false alarm at the FC; in the second (dual) problem we minimize the probability of transmission, subject to constraints on the probabilities of miss detection and false alarm at the FC. Based on the expressions of the objective functions and the constraints in each problem, we propose different optimization techniques  to address these two problems. 
Through analysis and simulations, we explore and provide the conditions (in terms of communication channel signal-to-noise ratio, degree of correlation among sensor observation noises, and maximum allowed false alarm probability) under which CRT schemes outperform pure censoring scheme.
%
%
\end{abstract}
\section{Introduction}
One of the important  wireless sensor network (WSN) applications is distributed binary detection, where  battery-powered wireless sensors are deployed over a sensing field to detect the presence or
absence of a target. 
Classical distributed detection \cite{c7,c8, c9, c10} is a powerful theoretical framework that enables system-level designers to formulate and address various problems pertaining to WSNs used for distributed detection. 
%
%
%
%
Motivated by the key observation that, when detecting a rare event 
transmitting many ``0'' decisions or low informative observation is wasteful in terms of communication cost, \cite{c11} introduced the idea of censoring, where sensors censor their ``uninformative'' observations and only transmit their ``informative'' observations. \cite{c11} showed that for conditionally independent sensor observations 
and under a communication rate constraint, a sensor should transmit  its (quantized) local likelihood  ratio (LLR) to FC only if it lies outside a certain single interval (so-called ``no-send'' interval). 
%
%
%
%
Leveraging on the results in \cite{c11}, the authors in \cite{c12} considered the extreme quantization case where a sensor transmits only one bit (sends ``1'') when its LLR exceeds a given threshold and remains silent when its LLR is below that threshold. Such a censoring scheme is effectively an on-off keying (OOK) signaling. With this OOK signaling, \cite{c12} incorporated the effects of wireless fading channels via developing (sub-)optimal fusion rules at the FC.
%
%
%
%
%
Rather than partitioning the LLR domain into two disjoint ``no-send'' and ``send'' intervals and using OOK signaling for wireless transmission as in\cite{c12}, the authors in \cite{c13, c14} proposed another censoring scheme, in which ``send'' interval is further divided into two intervals to increase the amount of information transmitted to the FC.  
%
%
%
%
Censoring sensors has also been investigated for spectrum sensing in cognitive radios \cite{c15, c16, c17}, albeit for conditionally independent observations.

On the other hand, the concept of ``randomized quantizer'' for an NP distributed detection problem was first introduced in \cite{c8, c9}, for conditionally independent observations. Unlike ``deterministic quantizer'' $\gamma$ (which maps a sensor's LLR to a discrete value according to a single local decision rule), a ``randomized quantizer'' chooses at random one local decision rule from a set of rules 
(with probability  $\mu_n$) 
for mapping a sensor's LLR to discrete values. 
Note that the quantizers in  \cite{c7,c10, c11, c12,c18, c19, c13,c14,c15, c16, c17} fall into the category of deterministic quantizers. 
The idea of  combining ``censoring'' and ``randomization'' was first introduced in \cite{c20,c21}.  The authors in \cite{c20} formulated the problem of finding optimal local decision rules  from Bayesian and NP viewpoints for conditionally independent observations, under communication rate constraint, and showed that likelihood-ratio-based local detectors are optimal.  The results in \cite{c20}  indicate that the effectiveness of independent randomization in choosing local decision rule, in terms of improving detection performance,  depends on whether or not sensors quantize their observations before transmission. 
%
%
%
%
 \cite{c22} provided a new framework for distributed detection with conditionally dependent observations (albeit without communication rate constraint and randomization in choosing local decision rule) that builds on a hierarchical conditional independent model and enabled the authors to formulate and address the problem of finding optimal local decision rules from Bayesian viewpoint. 
%
%
%

{\bf Our Contributions}: We consider an NP distributed binary detection problem where the FC is tasked with detecting a known signal in {\it correlated} Gaussian noises, using received signals from $K$ sensors. Our signal model is different from \cite{c11, c12,c18, c19, c13,c14,c15, c16, c17, c20,c21}, that considered conditionally  independent observations, i.e., in our setup sensors' observations conditioned on each hypothesis are {\it dependent}. 
With this signal model, our goal is to investigate whether or not {\it randomized transmission} can improve detection performance, under communication rate constraint.
To achieve this goal, we propose two novel schemes that combine the concepts of {\it censoring} and {\it randomized transmission} (which we call CRT-I and CRT-II schemes) and compare them with {\it pure censoring} scheme.
%
%
%
%
Assuming sensors transmit their non-zero symbols over orthogonal wireless fading channels,
let $P_M$ and $P_F$ be the probabilities of miss detection and false alarm at the FC, respectively, corresponding to the optimal likelihood-ratio test (LRT) fusion rule, and $P_t$ be the probability of transmission\footnote{We adopt this definition from \cite{c11,c20}, which have used this probability to measure communication rate, in the context of censoring sensors. 
%
%
%
For homogeneous sensors with identical observation distributions, a constraint on $P_t$ is equivalent to the communication rate constraint in \cite{c11,c20}.} under the null hypothesis only (i.e., signal is absent). 
We formulate two system-level constrained optimization problems, problem $(\mathcal{O})$ and  its dual problem $(\mathcal{S})$, for each scheme.  In problem $(\mathcal{O})$, we minimize $P_M$ subject to constraints on $P_t$ and $P_F$. In problem $(\mathcal{S})$, we minimize $P_t$ subject to constraints on $P_M$ and $P_F$.
For CRT schemes, 
we provide new optimization techniques to find the optimal randomization parameters $g,f$ as well as the FC threshold. 
To address problems $(\mathcal{O})$ and $(\mathcal{S})$ for CRT-I scheme, we first decompose each problem into two sub-problems and use some approximation 
to convert $P_M,P_F$ expressions 
into polynomial functions of $g,f$, and then solve a set of Karush-Kuhn-Tucker (KKT) conditions to find sub-optimal solutions. 
%
%
Different from problem $(\mathcal{O})$, however, in problem $(\mathcal{S})$ one of the sub-problems cannot be turned into a convex problem and hence we find a geometric programming approximation of that sub-problem and obtain sub-optimal solutions  to problem $(\mathcal{S})$.  
Similarly, we address problems $(\mathcal{O})$ and $(\mathcal{S})$ for CRT-II scheme, with the difference that $P_M,P_F$ expressions are polynomial functions of  $g,f$.
Based on our analytical solutions 
we provide the conditions 
under which our proposed CRT schemes outperform pure censoring scheme and explore numerically the deteriorating effects of incorrect correlation information (correlation mismatch) on the detection performance.
While independent randomization strategy cannot improve detection performance when sensors are restricted to transmit discrete values, for conditionally independent observations
\cite{c21}, our results show that this conclusion changes for conditionally dependent observations, and one can improve detection performance using our simple and easy-to-implement CRT schemes.

Our problem formulation and setup are different  from the related literature in the following aspects. Different from \cite{c20,c21,c22} that find
the forms of the optimal local decision rules, we fix the form and focus on finding the optimal randomization parameters, to show that randomized transmission can improve detection performance, when sensors' observations are conditionally dependent. Also, the bandwidth constrained communication channels between sensors and FC in  \cite{c20,c21,c22} are modeled error-free, whereas we consider wireless fading channels. 
Although \cite{c13} maps a sensor's observation into a ternary transmit symbol $u_k \in \{-1,0,1\}$, there is no randomized transmission, the communication channels are modeled as unfaded Gaussian channels, and most importantly, sensors' observations are conditionally independent.
%
%
\begin{figure}[!htbp]
\centering
\includegraphics[scale=0.28]{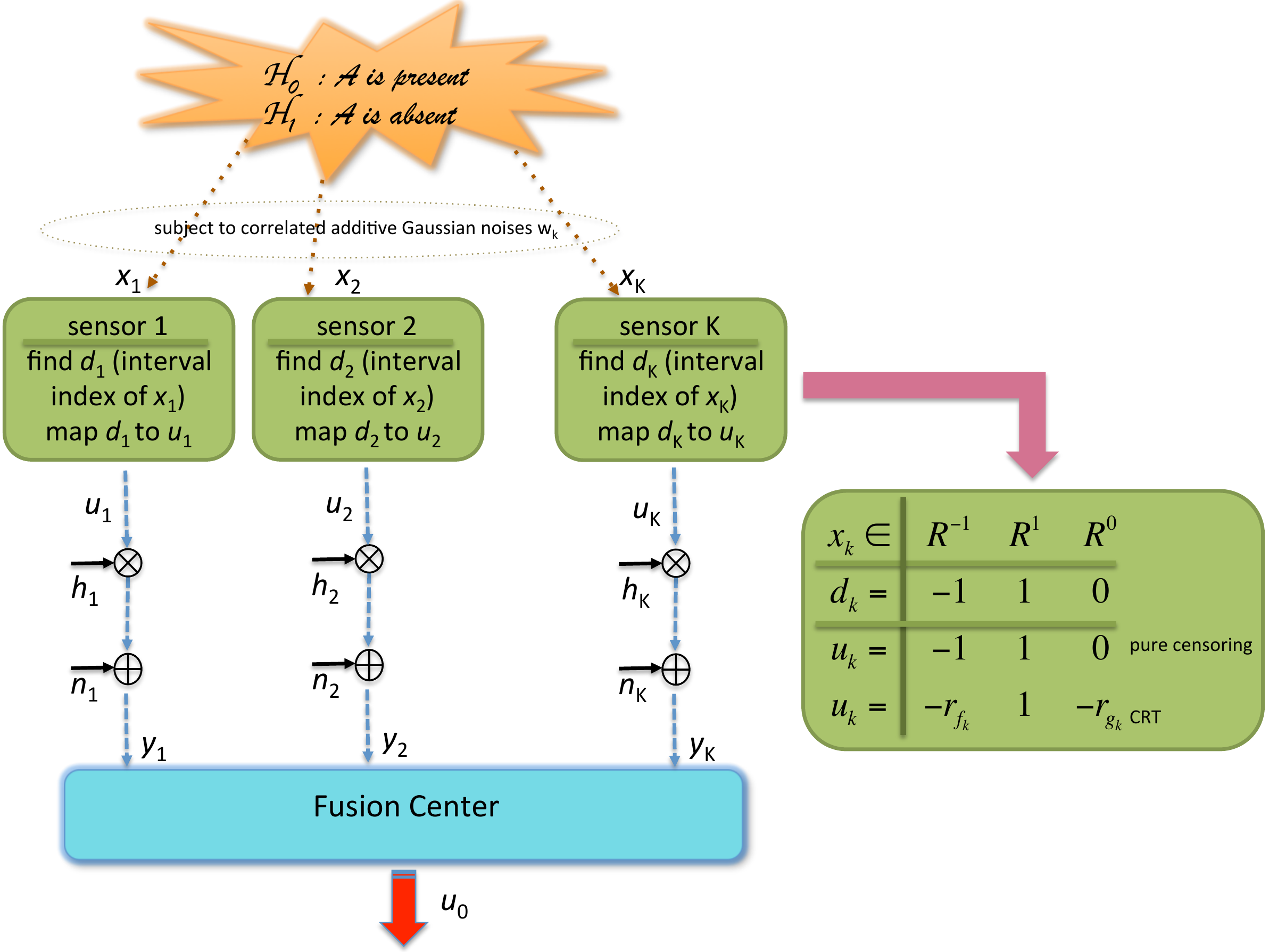}
\label{our-system-model-fig}
\caption{Our system model consists of $K$ sensors and a FC. Each sensor maps its real-valued observation $x_k$ into a ternary tranmsit symbol $u_k\in \{-1,0,1\}$.  Non-zero $u_k$'s are sent over orthogonal fading chanenls. The FC decides on the underlying binary hypothesis via fusing $y_k$'s using LRT in (\ref{u_0_opt}).}
\vspace{-0.6cm}
\end{figure}
\section{System Model and Problem Formulation}\label{System_model}
We consider the binary hypothesis testing problem of detecting a known signal ${\cal A}$ in correlated Gaussian noises, based on observations of $K$ distributed homogeneous sensors.
The FC is tasked with determining whether the unknown hypothesis is $\mathcal{H}_{1}$ or $\mathcal{H}_{0}$ (i.e., whether the signal is present or not), via fusing the received signals from $K$ sensors. Let $x_k$ denote observation of sensor ${\cal S}_k$. Our signal model is
\begin{equation}\label{original-binary-hypothesis}
\mathcal{H}_{0} :~~x_{k} = w_{k},~~
\mathcal{H}_{1} :~~x_{k} = {\cal A}+ w_{k},~~\mbox{for}~~k=1,\ldots,K.
\end{equation}
We assume noises $w_k$'s are {\it dependent} and identically distributed Gaussian random variables, that is $w_k \sim {\cal N}(0, \sigma_w^2)$ with covariance $\rho\sigma_w^2$, where $\rho$ is the correlation coefficient \cite{c23}.
Suppose ${\cal S}_k$ partitions its observation space\footnote{The choice of partitioning the observation space at each sensor into three disjoint intervals resembles the one in \cite{c13}, with the difference that in general $\tau_2  \neq - \tau_1$. Our choice is motivated by the result in \cite{c22} which states that, for conditionally dependent observations, if each sensor is restricted to map its observation to one of three discrete values, there exists one two-threshold quantizer at each sensor that minimizes the error probability.}
%
 into three intervals  $\mathcal{R}^{-1}=(-\infty,\tau_2)$, $\mathcal{R}^0=[\tau_2,\tau_1]$, $\mathcal{R}^1=(\tau_1,\infty)$. Upon making an observation, ${\cal S}_k$ finds the interval index $d_k$ corresponding to $x_k$, where $d_k \in \{-1,0,1\}$.
Next, ${\cal S}_k$  maps the interval index $d_k$ to a ternary transmit symbol $u_k \in \{-1,0,1\}$, where $u_k=-1,1,0$ correspond to ${\cal S}_k$ sends $-1$, ${\cal S}_k$ sends $1$, and ${\cal S}_k$ does not send and remains silent (${\cal S}_k$ censors), respectively.
Symbols $u_k$'s are transmitted over orthogonal wireless fading channels to the FC, subject to additive white Gaussian noise (AWGN).
The received signal at the FC from ${\cal S}_k$ is (see Fig. 1)
\begin{eqnarray}\label{y-received-at-FC}
&y_k= u_k h_k + v_k~\\
&\mbox{where}~~h_k \sim {\cal C}{\cal N}(0, \sigma_{h}^2),~~v_k \sim {\cal C}{\cal N}(0, \sigma_{v}^2),~~\mbox{for}~~k=1,\ldots,K,\nonumber
\end{eqnarray}
and $h_k$ represents the fading channel coefficient corresponding to the channel between  ${\cal S}_k$ and the FC, and $v_k$ denotes the AWGN. For coherent reception at the FC, the optimal fusion rule is likelihood-ratio test (LRT) as the following
\begin{eqnarray}\label{u_0_opt}
\gamma_0(y_1,y_2,\ldots,y_K)=u_0=\begin{cases}
1~~\mbox{if }\frac{f(y_1,y_2,\ldots,y_K|\mathcal{H}_1)}{f(y_1,y_2,\ldots,y_K|\mathcal{H}_0)}>t\\
0~~\mbox{if }\frac{f(y_1,y_2,\ldots,y_K|\mathcal{H}_1)}{f(y_1,y_2,\ldots,y_K|\mathcal{H}_0)}\leq t
\end{cases}
\end{eqnarray}
where $u_0=1$ ($u_0=0$) indicates the FC decides the signal is (not) present, $f(y_1,y_2,\ldots,y_K|\mathcal{H}_m)$ denotes the joint probability density function (pdf) of the received signals at the FC under hypothesis $\mathcal{H}_m$, $m=0,1$ and $t$ is the FC threshold.
%
To explore the effectiveness of randomized transmission, we propose two novel schemes that combine the concepts of censoring and randomized transmission (CRT-I and CRT-II schemes) and compare their performance (in terms of the reliability of the final decision $u_0$ at the FC and transmission rate) against that of {pure censoring} scheme.
Specifically,  in pure censoring scheme we let $u_k=d_k$ at ${\cal S}_k$. In the two proposed CRT schemes,  we introduce randomization when mapping interval index $d_k$ to transmit symbol $u_k$ at ${\cal S}_k$ as the following
\begin{eqnarray}\label{u_k_opt}
u_k=\begin{cases}
d_k~~~\mbox{if }x_k \in {\cal R}^{1}\\
r_{g_k}(d_k-1)~~~\mbox{if }x_k \in {\cal R}^{0}\\
r_{f_k}d_k~~~\mbox{if }x_k \in {\cal R}^{-1}
\end{cases}
\end{eqnarray}
in which $r_{g_k}, r_{f_k} \in \{0,1\}$ for $k=1,...,K$ are different realizations of two independent Bernoulli random variables with parameters $ 0 \leq g,f \leq 1$ ($g,f$ will be optimized). The two CRT schemes are different in the following way:  in CRT-I scheme sensors know $g, f$. Sensor ${\cal S}_k$ (independent of other sensors)
generates $r_{g_k}, r_{f_k}$ and uses these values to map $d_k$ to $u_k$, however, the FC is unaware of the values $\{r_{g_k}, r_{f_k}\}_{k=1}^K$  employed at the sensors. In CRT-II sensors do not know $g,f$. The FC generates  $\{r_{g_k}, r_{f_k}\}_{k=1}^K$ for all sensors, independent of each other,  
and informs each sensor of the values
%
%
 that should be employed for mapping $d_k$ to $u_k$.
%
%
%
%
We define probability of censoring (i.e., no transmission) as $P_c=P(u_k=0|\mathcal{H}_0)$ and probability of transmission as $P_t=P(u_k=1|\mathcal{H}_{0})+P(u_k=-1|\mathcal{H}_{0})$.
For pure censoring scheme $P_c=P(x_k \in {\cal R}^{0}|\mathcal{H}_0)$ and $P_t=P(x_k \in {\cal R}^{1}|\mathcal{H}_{0})+ P(x_k \in {\cal R}^{-1}|\mathcal{H}_{0})=1- P(x_k \in {\cal R}^{0}|\mathcal{H}_{0})$.
For CRT-I and CRT-II 
\begin{eqnarray}
&&P_c
=(1-g)P(x_k \in {\cal R}^{0}|\mathcal{H}_0)+(1-f)P(x_k \in {\cal R}^{-1}|\mathcal{H}_0), \label{definition-P-t} \\
&&P_t
= P(x_k \in {\cal R}^{1}|\mathcal{H}_0)+ g P(x_k \in {\cal R}^{0}|\mathcal{H}_0)+ f P(x_k \in {\cal R}^{-1}|\mathcal{H}_0).\nonumber
\end{eqnarray}
Note that for all three schemes (pure censoring,  CRT-I and CRT-II) we have $P_c+P_t=1$.
Let $P_M=P(u_0=0|\mathcal{H}_{1})$ and $P_F=P(u_0=1|\mathcal{H}_{0})$ denote the probabilities of miss detection and false alarm at the FC, respectively. 
We consider two system-level constrained optimization problems for each scheme.  In the first problem $(\mathcal{O})$, we minimize $P_M$ subject to constraints on $P_t$ and $P_F$. In the second problem $(\mathcal{S})$, we minimize $P_t$ subject to constraints on $P_M$ and $P_F$, i.e.,
\begin{eqnarray}\label{min_det}
\begin{array}{cc}
 \begin{array}{cc}
&\min~~P_{M}~~(\mathcal{O})\\
&\text{s.t.}~~
P_t =p_0,~~
P_F\leq \beta
    \end{array}
 &,\begin{array}{cc}
&\min~~P_t~~(\mathcal{S})\\
&\text{s.t.}~~
P_M\leq \alpha,~~
P_F\leq \beta
   \end{array}
\end{array}
\end{eqnarray}
where $\alpha, \beta, p_0$ are the largest tolerable $P_M, P_F$ and the maximum $P_t$, respectively.
For pure censoring scheme, we let the optimization variables be the local thresholds $\tau_1,\tau_2$ and the FC threshold $t$. 
For CRT-I and CRT-II,  we let the optimization variables be the randomization  parameters $g,f$ and the FC threshold $t$ (assuming sensors use the same local thresholds $\tau_1,\tau_2$ as for pure censoring scheme). 
Section \ref{derive PF_PM} derives $P_M,P_F$ expressions for the three schemes. Sections \ref{solve-O-CRT-I} and \ref{solve-O-CRT-II}  address problem $(\mathcal{O})$ in (\ref{min_det}) for CRT-I and CRT-II schemes, respectively. Sections \ref{solve-S-random-i} and \ref{solve-S-random-ii}  address problem $(\mathcal{S})$ in (\ref{min_det}) for CRT-I and CRT-II schemes, respectively.
The solutions to these problems provide us insights on the effectiveness of CRT schemes (with respect to pure censoring), when sensors' observations are conditionally dependent.





\section{Deriving $P_M$ and $P_F$ Expressions}\label{derive PF_PM}

For pure censoring scheme $P_M, P_F$  depend on $\tau_1,\tau_2, t$ and for CRT schemes they depend on $\tau_1,\tau_2, g, f, t$. In the following, we  derive $P_M, P_F$ for CRT-I and CRT-II schemes.
%
%
When we let $g=0, f=1$ into $P_M, P_F$ expressions of either CRT schemes, we reach $P_M,P_F$ expressions of pure censoring scheme.
%

\subsection{CRT-I Scheme}\label{derive PF_PM_opt_i}
To characterize $P_M, P_F$ we need the following definitions. For each sensor $\mathcal{S}_k$ we define row vector $\boldsymbol{c}_k=[i, u_k]$, where index $i$ indicates the interval which $x_k$ belongs to, i.e., $x_k \in \mathcal{R}^{i}$ for $i \in \{-1,0,1\}$, and $u_k$ is the transmitted symbol, i.e., $u_k \in \{-1,0,1\}.$  We define $K \times 2$ matrix $\boldsymbol{C}=[\boldsymbol{c}_1;\ldots;\boldsymbol{c}_K]$, whose rows are vectors $\boldsymbol{c}_k, ~k=1,...,K$. For CRT-I scheme and the above definitions, we recognize the non-empty set of sensors' indices fall into 5 categories
${\cal K}_1=\{k|\boldsymbol{c}_k=[1, 1]=\boldsymbol{c}^1\},~{\cal K}_2=\{k|\boldsymbol{c}_k=[0, 0]=\boldsymbol{c}^2\},~{\cal K}_3=\{k|\boldsymbol{c}_k=[0, -1]=\boldsymbol{c}^3\},
~{\cal K}_4=\{k|\boldsymbol{c}_k=[-1, 0]=\boldsymbol{c}^4\},~{\cal K}_5=\{k|\boldsymbol{c}_k=[-1, -1]=\boldsymbol{c}^5\}$.
%
%
We define row vector $\boldsymbol{a}=[a_1,\ldots,a_5]$, where its entries are $a_l=|{\cal K}_l|$, $a_l \in \{0,\ldots,K\}$ and satisfy $\sum_{l=1}^5 a_l=K$. We define $K \times 2$ matrix $\boldsymbol{C}^{a}$ such that the first $a_1$ rows are $\boldsymbol{c}^1$, the next $a_2$ rows are $\boldsymbol{c}^2$,  the next $a_3$ rows are $\boldsymbol{c}^3$, the next $a_4$ rows are $\boldsymbol{c}^4$, and the last $a_5$ rows are $\boldsymbol{c}^5$. Let $\boldsymbol{C}^{a}(k,1)$ and $\boldsymbol{C}^{a}(k,2)$ denote $(k,1)$-th and $(k,2)$-th entries of matrix $\boldsymbol{C}^{a}$, respectively.
Also, we define the following probabilities
\begin{eqnarray}
&&P_u(\tau_1,\tau_2, g, f, t, \boldsymbol{a})\!\!=\!\!P(u_0=1|\tau_1,\tau_2, g, f, t, \boldsymbol{C}=\boldsymbol{C}^{a}), \label{definition P_x_opt_i} \\
&&P_{xm}(\tau_1,\tau_2,\boldsymbol{a}) \!\!=\!\! P(x_k \in \mathcal{R}^{\boldsymbol{C}^{a}(k,1)}~\forall k|\tau_1,\tau_2,{\mathcal H}_m)~m=0,1\nonumber
\end{eqnarray}
Since noises $w_k$'s are correlated $P_{xm}$ cannot be decoupled across sensors 
and depends on the correlation coefficient $\rho$.
Noting that sensors are homogeneous, 
using the definitions in (\ref{definition P_x_opt_i}) and
the fact that, given the intervals to which $x_k$ belongs to, symbols $u_k$'s are conditionally independent,
we can express $P_M$ and $P_F$ in terms of $\tau_1,\tau_2, g, f, t$ as in (\ref{definition_P_ND_opt_i}).
\begin{figure*}[!htbp]
\begin{eqnarray}\label{definition_P_ND_opt_i} 
&&P_M= \sum_{a_2,...,a_5} \underbrace{\frac{K!(1-P_u(\tau_1,\tau_2, g, f, t, \boldsymbol{a})) P_{x1}(\tau_1,\tau_2,\boldsymbol{a})}{(K-a_2-a_3-a_4-a_5)!a_2!...,a_5!}
}_{=c_{a_2,a_3,a_4,a_5}} (1-g)^{a_2}  g^{a_3}  (1-f)^{a_4} f^{a_5} \nonumber\\
&&P_{F}=\sum_{a_2,...,a_5} \underbrace{\frac{K!P_u(\tau_1,\tau_2, g, f, t,\boldsymbol{a})P_{x0}(\tau_1,\tau_2,\boldsymbol{a})}{(K-a_2-a_3-a_4-a_5)!a_2!...,a_5!}}_{=d_{a_2,a_3,a_4,a_5}}  (1-g)^{a_2}  g^{a_3} (1-f)^{a_4} f^{a_5}\\
&&1-P_u(\tau_1,\tau_2, g, f, t,\boldsymbol{a})=\int_{\{y_k,h_k\}_{k=1}^K} \mathbbm{1}_{\{0\}}(u_0|\{y_k,h_k\}_{k},\tau_1,\tau_2, g, f, t)\prod_{k=1}^K f(y_k|u_k=\boldsymbol{C}^{a}(k,2),h_k)f(h_k) dh_1...dh_K dy_1...dy_K \label{Int P_u_opt_i}
\end{eqnarray}
\end{figure*}

In the following, we focus on $P_u(\tau_1,\tau_2, g, f, t,\boldsymbol{a})$ in $P_M, P_F$ 
and show that this probability is a non-polynomial function of $g,f$.  Note that $u_0$ depends on the communication channels $h_k$ between sensors and the FC. Hence, 
one needs to take average over all realizations of $h_k, \forall k$. Since $[u_1,\ldots,u_K]\rightarrow [y_1,\ldots,y_K] \rightarrow u_0$ forms a Markov chain, we reach (\ref{Int P_u_opt_i}).
The indicator function in (\ref{Int P_u_opt_i}) is defined as $\mathbbm{1}_A(x)=1,~\text{if~} x \in A$ and $\mathbbm{1}_A(x)=0, ~\text{if~} x  \notin A$. Let examine the terms in (\ref{Int P_u_opt_i}). The marginal pdf $f(h_k)$ is known since $h_k \sim {\cal C}{\cal N}(0, \sigma_{h}^2)$.
Conditioned on $u_k=\boldsymbol{C}^{a}(k,2)$ and $h_k$,  $y_k$ is Gaussian with mean $\boldsymbol{C}^{a}(k,2)h_k$ and variance $\sigma^2_v$. 
\begin{figure*}[!htbp]
	\scriptsize
\begin{align}\label{u_0_detail_opt_i}
\!\!\!\!\!\! u_0=\mathbbm{1}_{\mathbbm{R}^+}(\frac{ f(y_1,\ldots,y_K|\{h_k\}_{k=1}^K,\tau_1,\tau_2, g, f,\mathcal{H}_1)}{f(y_1,\ldots,y_K|\{h_k\}_{k=1}^K,\tau_1,\tau_2, g, f,\mathcal{H}_0)}-t)
=\mathbbm{1}_{\mathbbm{R}^+}(\frac{\sum_{v_1=1}^5...\sum_{v_K=1}^5\prod_{k=1}^K f(y_k|u_k=\boldsymbol{c}^{v_k}(2),h_k)P_{x1}(\tau_1,\tau_2,\boldsymbol{a})  (1-g)^{a_2} g^{a_3} (1-f)^{a_4} f^{a_5}}
{\sum_{v_1=1}^5...\sum_{v_K=1}^5\prod_{k=1}^K f(y_k|u_k=\boldsymbol{c}^{v_k}(2),h_k)P_{x0}(\tau_1,\tau_2,\boldsymbol{a}) (1-g)^{a_2} g^{a_3}  (1-f)^{a_4} f^{a_5} }-t)
\end{align}
\hrulefill
\end{figure*}
\normalsize
Next, we consider $\mathbbm{1}_{\{0\}}(u_0|\{h_k,y_k\}_{k=1}^K,\tau_1,\tau_2, g, f, t)$. Let $v_k \in \{1,2,3,4,5\}$ for $k=1,...,K$. For LRT fusion rule in (\ref{u_0_opt}) we can write (\ref{u_0_detail_opt_i}), where ${\boldsymbol{c}^{v_k}(2)}$ is the second entry of row vector $\boldsymbol{c}^{v_k}$.
%
From (\ref{u_0_detail_opt_i}) it is clear that the probability $P_u(\tau_1,\tau_2, g, f,t,\boldsymbol{a})$ in (\ref{definition_P_ND_opt_i}) and hence $P_M,P_F$ expressions for CRT-I scheme are non-polynomial functions of $g,f$. 

%
\subsection{CRT-II Scheme}\label{derive PF_PM_opt_ii}
To characterize $P_M,P_F$ we need the following definitions. For each sensor $\mathcal{S}_k$ we define row vector $\boldsymbol{c}_k=[i, r_{f_k}, r_{g_k}]$, where index $i$ indicates the interval which $x_k$ belongs to, i.e., $x_k \in \mathcal{R}^i$ for $i \in \{-1,0,1\}$, and $r_{f_k}, r_{g_k}$ are the realizations of independent Bernoulli random variables with parameters $f,g$. Note that given $i, r_{f_k}, r_{g_k}$, symbol $u_k$ is known and is not needed to be included in the definition of $\boldsymbol{c}_k$. We define $K \times 3$ matrix $\boldsymbol{C}=[\boldsymbol{c}_1;\ldots;\boldsymbol{c}_K]$, whose rows are vectors  $\boldsymbol{c}_k, ~k=1,...,K$. For CRT-II scheme and the above definitions, we recognize the non-empty set of sensors' indices fall into 12 categories
${\cal K}_{1_m}=\{k|\boldsymbol{c}_k=[1, 1, m]=\boldsymbol{c}^{1_m}\}, {\cal K}_{2_m}=\{k|\boldsymbol{c}_k=[0, m, 0]=\boldsymbol{c}^{2_m}\},  {\cal K}_{3_m}=\{k|\boldsymbol{c}_k=[0, m, 1]=\boldsymbol{c}^{3_m}\}, {\cal K}_{4_m}=\{k|\boldsymbol{c}_k=[-1, 0, m]=\boldsymbol{c}^{4_m}\},
 {\cal K}_{5_m}=\{k|\boldsymbol{c}_k=[-1, 1, m]=\boldsymbol{c}^{5_m}\}, {\cal K}_{6_m}=\{k|\boldsymbol{c}_k=[1, 0, m]=\boldsymbol{c}^{6_m}\}$ for $m=0,1$.
We define row vector $\boldsymbol{a}=[a_{1_1},a_{1_0},\ldots,a_{6_1},a_{6_0}]$ where its entries are $a_{l_m}=|{\cal K}_{l_m}|$, $a_{l_m} \in \{0,\ldots,K\}$ and satisfy $\sum_{l_m} a_{l_m}=K$. We define $K \times 3$ matrix $\boldsymbol{C}^{a}$ such that the first $a_{1_1}$ rows are $\boldsymbol{c}^{1_1}$, the next $a_{1_0}$ rows are $\boldsymbol{c}^{1_0}$ and so on, and the last $a_{6_0}$ rows are $\boldsymbol{c}^{6_0}$.
Let $\boldsymbol{C}^{a}(k,1)$, $\boldsymbol{C}^{a}(k,2)$, and $\boldsymbol{C}^{a}(k,3)$ denote $(k,1)$-th, $(k,2)$-th and $(k,3)$-th entries of matrix $\boldsymbol{C}^{a}$, respectively.
Also, we define the following probabilities
\begin{align}
&P_{u_0}(\tau_1,\tau_2, t, \boldsymbol{a})=P(u_0=0|\tau_1,\tau_2, t, \boldsymbol{C}=\boldsymbol{C}^{a}), ~\nonumber\\
&P_{u_1}(\tau_1,\tau_2, t, \boldsymbol{a})=P(u_0=1|\tau_1,\tau_2, t, \boldsymbol{C}=\boldsymbol{C}^{a})\label{definition P_x_opt_ii}
\end{align}
We let $a_g=\sum_{k}\mathbbm{1}_{\{1\}}(r_{g_k})=a_{1_1}+a_{3_0}+a_{3_1}+a_{4_1}+a_{5_1}+a_{6_1}$ and $a_f=\sum_{k}\mathbbm{1}_{\{1\}}(r_{f_k})=a_{1_0}+a_{1_1}+a_{2_1}+a_{3_1}+a_{5_0}+a_{5_1}$. 
Noting that 
sensors are homogeneous, and using  
the definitions in (\ref{definition P_x_opt_ii}) and (\ref{definition P_x_opt_i}), we can express $P_M, P_F$ in terms of $\tau_1,\tau_2, g, f,t$ as in (\ref{definition P_ND_opt_ii}).
\begin{figure*}[!htbp]
\begin{align}\label{definition P_ND_opt_ii}
&P_M = \sum_{a_{1_1}}\sum_{a_{1_0}}...\sum_{a_{6_1}}\sum_{a_{6_0}}  \frac{K! P_{u_0}(\tau_1,\tau_2,t, \boldsymbol{a}) P_{x1}(\tau_1,\tau_2, \boldsymbol{a}) }{a_{1_1}!a_{1_0}!...a_{6_1}!a_{6_0}!} (1-g)^{K-a_g} g^{a_g} (1-f)^{K-a_f}f^{a_f}  \\
&P_F=
\sum_{a_{1_1}}\sum_{a_{1_0}}...\sum_{a_{6_1}}\sum_{a_{6_0}} \frac{K!P_{u_1}(\tau_1,\tau_2,t, \boldsymbol{a})P_{x0}(\tau_1,\tau_2, \boldsymbol{a})}{a_{1_1}!a_{1_0}!...a_{6_1}!a_{6_0}!} (1-g)^{K-a_g} g^{a_g} (1-f)^{K-a_f}f^{a_f} \nonumber
\end{align}
\end{figure*}

In the following we focus on
$P_{u_0}(\tau_1,\tau_2,t, \boldsymbol{a})$ in $P_M,P_F$
and show that this probability is a polynomial function of $g,f$. Note that $u_0$ depends on the communication channels $h_k$ between sensors and the FC. Hence, 
one needs to take average over all realizations of $h_k, \forall k$.
Since $[u_1,\ldots,u_K]\rightarrow [y_1,\ldots,y_K] \rightarrow u_0$ forms a Markov chain, we reach (\ref{P_u_0_int_opt_ii}).
\begin{figure*}[!htbp]
\scriptsize
\begin{align}
&P_{u_0}(\tau_1,\tau_2, t, \boldsymbol{a})
=\int_{\{y_k,h_k\}_{k=1}^K} \mathbbm{1}_{\{0\}}(u_0|\{y_k,h_k\}_{k=1}^K,r_{f_k}=\boldsymbol{C}^{a}(k,2),r_{g_k}=\boldsymbol{C}^{a}(k,3)~\forall k,\tau_1,\tau_2,t )\prod_{k=1}^K f(y_k|u_k,h_k)f(h_k)dh_1...dh_Kdy_1...dy_K \label{P_u_0_int_opt_ii}\\
&u_0=\mathbbm{1}_{\mathbbm{R}^+}(\frac{f(y_1,\ldots,y_K|\{h_k\}_{k=1}^K,\boldsymbol{r}_f,\boldsymbol{r}_g,\tau_1,\tau_2,\mathcal{H}_1)}
{f(y_1,\ldots,y_K|\{h_k\}_{k=1}^K,\boldsymbol{r}_f,\boldsymbol{r}_g,\tau_1,\tau_2,\mathcal{H}_0)}-t)=\mathbbm{1}_{\mathbbm{R}^+}(\frac{\sum_{\nu_1={1_1}}^{6_0}...\sum_{\nu_K={1_1}}^{6_0}\prod_{k=1}^K \big(f(y_k|u_k,h_k)\mathbbm{1}_{\{r_{f_k}\}}(\boldsymbol{C}^{a}(k,2))\mathbbm{1}_{\{r_{g_k}\}}(\boldsymbol{C}^{a}(k,3))\big)P_{x1}(\boldsymbol{a},\tau_1,\tau_2)
}{\sum_{\nu_1={1_1}}^{6_0}...\sum_{\nu_K={1_1}}^{6_0}
\prod_{k=1}^K \big(f(y_k|u_k,h_k)\mathbbm{1}_{\{r_{f_k}\}}(\boldsymbol{C}^{a}(k,2))\mathbbm{1}_{\{r_{g_k}\}}(\boldsymbol{C}^{a}(k,3))\big)P_{x0}(\boldsymbol{a},\tau_1,\tau_2)}-t)
\label{u_0_detail_opt_ii} \end{align} 
\hrulefill
\end{figure*}
\normalsize
Let examine the terms in (\ref{P_u_0_int_opt_ii}). 
Conditioned on $u_k$ (which is determined by $\boldsymbol{C}^a$) and $h_k$, $y_k$ is Gaussian with mean $ u_k h_k$ and variance $\sigma^2_v$. 
Next, we consider $\mathbbm{1}_{\{0\}}(u_0|\{h_k,y_k\}_{k=1}^K, \boldsymbol{r}_f,\boldsymbol{r}_g, \tau_1,\tau_2,t)$, where the vectors $\boldsymbol{r}_f=[r_{f_1},...,r_{f_K}],\boldsymbol{r}_g=[r_{g_1},...,r_{g_K}]$. Let $\nu_k \in \{1_1,1_0,\ldots,6_1,6_0\}$ for $k=1,...,K$. For LRT fusion rule in  (\ref{u_0_opt}) we can write (\ref{u_0_detail_opt_ii}).
%
%
%
%
%
Examining (\ref{definition P_ND_opt_ii}) we realize that $P_M, P_F$ expressions for CRT-II scheme are polynomial functions of $g,f$.
We note that, although in CRT-II scheme the FC is aware of $\{r_{g_k}, r_{f_k}\}_{k=1}^K$, $P_M, P_F$ expressions do not depend on these specific realizations and depend on $g,f$, since we effectively take average over these realizations. 


\section{Addressing Problem $(\mathcal{O})$}\label{min P_ND cons P_F}

Let start with problem $(\mathcal{O})$ in (\ref{min_det}) for pure censoring scheme with the optimization variables $\tau_1, \tau_2, t$. For conditionally independent observations, this constrained optimization problem was discussed in \cite{c11}. The authors in \cite{c11} noted that this  problem is not necessarily convex and local minima may be found. Let $\tau_1^d,\tau_2^d, t^d$ denote the solutions to problem $(\mathcal{O})$ in (\ref{min_det}) for pure censoring scheme. Our procedure to find these solutions is similar to the one in \cite{c11}, albeit with $P_M,P_F$ expressions derived in Section \ref{derive PF_PM}, which are functions of $\tau_1, \tau_2, t$. 
Our contribution in this section is addressing problem $(\mathcal{O})$ in (\ref{min_det}) for our proposed CRT schemes.
In Section \ref{derive PF_PM} we derived $P_M, P_F$ expressions for CRT schemes in terms of $\tau_1,\tau_2, g, f,t$. Assuming sensors partition their observation spaces into the same intervals regardless of the employed scheme, 
in this section we let $\tau_1=\tau_1^d,\tau_2=\tau_2^d$  and view $P_M,P_F$ expressions for CRT schemes as functions of $g,f,t$ only.

Consider problem $(\mathcal{O})$ in (\ref{min_det}) for CRT schemes with the optimization variables $g,f,t$ where $0 \leq g,f \leq 1$.
Using (\ref{definition-P-t}), we let $P_t=p_0$ and solve for $g$ to obtain
$g=\frac{p_0-P(x_k \in \mathcal{R}^1|\mathcal{H}_0)-f P(x_k \in \mathcal{R}^{-1}|\mathcal{H}_0)}{P(x_k \in \mathcal{R}^0|\mathcal{H}_0)}$.
Substituting this solution into the constraint $0\leq g \leq 1$ we reach the equivalent problem
\begin{align}\label{O_i_and_O_ii}
&\min_{t,f}~~P_{M}(t,f)~~(\mathcal{O'})\nonumber\\
&\text{s.t.}~~P_F(t,f)\leq \beta,~~
l'_0=\max(0,l_0)\leq f \leq \min(1,l_1)=l'_1
\end{align}
where $l_0=\frac{p_0-1+P(x_k \in \mathcal{R}^{-1}|\mathcal{H}_0)}{P(x_k \in \mathcal{R}^{-1}|\mathcal{H}_0)}$ and $l_1=\frac{p_0-P(x_k \in \mathcal{R}^1|\mathcal{H}_0)}{P(x_k \in \mathcal{R}^{-1}|\mathcal{H}_0)}$.
For the rest of this section, suppose $t_i^{opt},f_i^{opt}$ are the solutions to $(\mathcal{O}')$ for CRT-I scheme and $t_{ii}^{opt},f_{ii}^{opt}$ are the solutions to $(\mathcal{O}')$ for CRT-II scheme, where $t_{i}^{opt},t_{ii}^{opt}$ can be different from $t^d$ (i.e., the solution to problem $(\mathcal{O})$ in (\ref{min_det}) for pure censoring scheme).
To solve $(\mathcal{O}')$, we decompose it into two subproblems $(\mathcal{O}'_1)$ and $(\mathcal{O}'_2)$ as the following, and solve them in a sequential order without an iteration between them
\begin{eqnarray}\label{define-O-prime}
&&\mbox{given }~~t,~~\min_{f}~~P_{M}(t,f)~~(\mathcal{O}'_1) \nonumber\\
&&\text{s.t.}~~P_F(t,f)\leq \beta, ~~l'_0\leq f \leq l'_1 \nonumber\\
&&\mbox{given }f,~~\min_{t}~~P_{M}(t,f)~~(\mathcal{O}'_2)\nonumber\\
&&\text{s.t.}~~P_F(t,f)\leq \beta 
\end{eqnarray}
%

%
\subsection{CRT-I Scheme}\label{approach-i-solve-O}\label{solve-O-CRT-I}
Let start with  $P_M,P_F$ expressions in (\ref{definition_P_ND_opt_i}).
We simplify the notations by dropping the thresholds $\tau_1,\tau_2$ from $P_u(\tau_1,\tau_2, g, f,t,\boldsymbol{a})$ and $P_{xm}(\tau_1,\tau_2,\boldsymbol{a})$ and denoting these terms instead as $P_u(g,f,t,\boldsymbol{a})$ and $P_{xm}(\boldsymbol{a})$, since the thresholds are fixed at $\tau_1^d,\tau_2^d$.
Recall the terms  $c_{a_2,a_3,a_4,a_5}$ and $d_{a_2,a_3,a_4,a_5}$ in (\ref{definition_P_ND_opt_i}) depend on $g,f$ through the probability $P_u(g,f,t,\boldsymbol{a})$ and this probability is a non-polynomial function of $g,f$. This probability is characterized by how the FC incorporates its knowledge of $g,f$ in constructing its fusion rule.
%
Motivated by the fact that, there are efficient algorithms for solving polynomials that converge to their roots, we make the following assumption to reduce  $P_M,P_F$ expressions to two polynomial functions of $g,f$. We assume the FC ignores its knowledge of $g,f$ value in constructing its fusion rule.
This assumption becomes equivalent to letting $g=0,f=1$ in $P_u(g,f,t,\boldsymbol{a})$, without affecting other parts of $P_M,P_F$ expressions.
Under this assumption $P_M, P_F$ expressions in (\ref{definition_P_ND_opt_i}) reduce to $P_M', P_F'$ given in (\ref{P'_M,P'_F}).
\begin{figure*}[!htbp]
\scriptsize
\begin{align}\label{P'_M,P'_F}
P_M' = \sum_{a_2, a_3, a_4, a_5}\frac{K!(1-P_u(0,1,t, \boldsymbol{a}))P_{x1}(\boldsymbol{a})}{(K-a_2-a_3-a_4-a_5)!a_2!...a_5!}
(1-g)^{a_2} g^{a_3} (1-f)^{a_4} f^{a_5},~
P_F' =  \sum_{a_2, a_3, a_4, a_5} \frac{K!P_u(0,1,t,\boldsymbol{a})P_{x0}(\boldsymbol{a})}{(K-a_2-a_3-a_4-a_5)!a_2!...a_5!}  (1-g)^{a_2}  g^{a_3} (1-f)^{a_4} f^{a_5}
\end{align}
\hrulefill
\end{figure*}
\normalsize
Note $P_M', P_F'$ are now polynomial functions of $f$, assuming that $g$ in (\ref{P'_M,P'_F}) is substituted with its solution  in terms of $f$ given earlier.  Let $(\mathcal{O}''_1), (\mathcal{O}''_2)$ be similar to $(\mathcal{O}'_1), (\mathcal{O}'_2)$ in (\ref{define-O-prime}), with the difference that $P_M,P_F$ are now replaced with $P'_M,P'_F$ in (\ref{P'_M,P'_F}) and suppose $f^*_i, t^{*}_i$ are the solutions to $(\mathcal{O}''_1), (\mathcal{O}''_2)$, respectively.
We find $f^*_i, t^{*}_i$ as the following. 

$\bullet$ {\textbf{Solving $(\mathcal{O}''_1)$:}} let  $t=t^d$ be the solution to problem $(\mathcal{O})$ in (\ref{min_det}) for pure censoring. To solve $(\mathcal{O}''_1)$ and find $f^{*}_i$,
we use
the Lagrange multiplier method, and solve the corresponding
Karush-Kuhn-Tucker (KKT) conditions in (\ref{KKT-O3-1})-(\ref{KKT-O3-3}). Let $\mathfrak{L}(f,\lambda,\mu_1,\mu_2)$ be the Lagrangian for $(\mathcal{O}''_1)$, where $\lambda,\mu_1, \mu_2$ are the Lagrange multipliers. The  KKT conditions are
\begin{align}
&\frac{d\mathfrak{L}}{df}=\frac{dP'_{M}(t^d,f)}{df}+\lambda\frac{dP'_{F}(t^d,f)}{df}+\mu_1-\mu_2=0\label{KKT-O3-1}\\
&\lambda(P'_{F}(t^d,f)-\beta)=0,~P'_{F}(t^d,f)\leq\beta,~\lambda \geq0\label{KKT-O3-2}\\
&\mu_1(f-l'_1)=0,~f\leq l'_1,~\mu_1\geq0, ~\mu_2(l'_0-f)=0,~l'_0\leq f,~\mu_2\geq0\label{KKT-O3-3}
\end{align}
$\bullet$ {\textbf{Solving $(\mathcal{O}''_2)$:}} 
Next, we let $f=f^{*}_i$ and solve $(\mathcal{O}''_2)$ to find $t^*_i$ such that the inequality constrain holds with equality, i.e., $P'_F(t_i^*,f_i^*)=\beta$. Note that $t_i^*$ can be different from $t^d$.
%


Although in general $t_i^{opt} \neq t_i^*, ~f_i^{opt} \neq f_i^*$, we prove in Lemma \ref{lemma-one}
 that if $0<f^{*}_i<1$ then $0<f_i^{opt}<1$, and hence CRT-I scheme is more effective than pure censoring scheme (with $f=1,g=0$), i.e., it provides a lower miss detection probability, under the same constraints on false alarm and transmission probabilities. Theorem \ref{dP_M_df_thm_opt_i} and Corollary \ref{cor-random1} identify the conditions under which we have $0< f^{*}_i<1$.

\begin{lem}\label{lemma-one}
Recall $f_i^{opt}$ is the solution to $(\mathcal{O}')$ and $f^{*}_i$ is the solution to $(\mathcal{O}''_1)$. If $0<f^{*}_i<1$ then $0<f_i^{opt}<1$.
\end{lem}

\begin{proof}
See Appendix \ref{appendix-proof-of-lemma1}.
\end{proof}

%
%

\begin{thm}\label{dP_M_df_thm_opt_i}
If correlation coefficient $\rho$ is sufficiently large such that observations $x_k, \forall k$ fall in two
consecutive intervals, we have
$\frac{dP'_{M}(t^d,f)}{df}|_{f=1}\approx 0$. On the other hand, for every $\rho$ we have $\frac{dP'_{F}(t^d,f)}{df}|_{f=1}>0$.
\end{thm}
\begin{proof}
See Appendix \ref{appendix-proof-dP_M_df_thm_opt_i}.
\end{proof}

\begin{cor}\label{cor-random1}
If correlation coefficient $\rho$ is sufficiently large such that observations $x_k, \forall k$ fall in two consecutive intervals, we have $0<f^{*}_i<1$.
\end{cor}
\begin{proof}
We consider the KKT conditions in (\ref{KKT-O3-1})-(\ref{KKT-O3-3}).
For $f=1$ from (\ref{KKT-O3-3}) we find $\mu_1>0$ and $\mu_2=0$. Also, for $f=1$ from (\ref{KKT-O3-2}) we have $\lambda \geq 0$.
Considering these values for $\mu_1, \mu_2, \lambda$ in (\ref{KKT-O3-1}) and $\frac{dP'_{M}(t^d,f)}{df}|_{f=1}\approx 0$, $\frac{dP'_{F}(t^d,f)}{df}|_{f=1}>0$ from Theorem \ref{dP_M_df_thm_opt_i} we reach $\frac{d\mathfrak{L}}{df}>0$.
Consequently, $f_i^*=1$ cannot be the solution of the KKT conditions, and thus $0<f^{*}_i<1$.
\end{proof}

%
\subsection{CRT-II Scheme}\label{solve-O-CRT-II}
Let start with $P_M,P_F$ expressions in (\ref{definition P_ND_opt_ii}). Different from Section \ref{approach-i-solve-O}, $P_M,P_F$ expressions are polynomial functions of $f$, assuming that $g$ in (\ref{definition P_ND_opt_ii}) is substituted with its solution  in terms of $f$.
Consider $(\mathcal{O}'_1), (\mathcal{O}'_2)$ in (\ref{define-O-prime}), with $P_M,P_F$ in (\ref{definition P_ND_opt_ii}). Suppose $f^*_{ii}, t^{*}_{ii}$ are the solutions to $(\mathcal{O}'_1), (\mathcal{O}'_2)$, respectively. We use the same procedures as Section \ref{approach-i-solve-O} to solve these two subproblems and find $t^{*}_{ii},f^*_{ii}$.
Although in general $t_{ii}^{opt} \neq t_{ii}^*, ~f_{ii}^{opt} \neq f_{ii}^*$, following similar arguments in Lemma \ref{lemma-one} of Section \ref{approach-i-solve-O}, one can prove that if $0<f^{*}_{ii}<1$ then $0<f_{ii}^{opt}<1$. This result implies that CRT-II scheme is more effective than pure censoring scheme, i.e.,  it provides a lower miss detection probability, under the same constraints on false alarm and transmission probabilities. Theorem \ref{dP_F_df_thm_opt_ii} and Corollary \ref{cor-random2} identify the conditions under which $0<f^{*}_{ii}<1$.

\begin{thm}\label{dP_F_df_thm_opt_ii}
If $P_t=p_0$ is sufficiently small such that $\tau_{2}^{d}<0$, for every $\rho$ we have $\frac{dP_{M}(t^d,f)}{df}|_{f=1}>0$ and $\frac{dP_{F}(t^d,f)}{df}|_{f=1}>0$.
\end{thm}
\begin{proof}
See Appendix \ref{appendix-proof-dP_F_df_thm_opt_ii}.
\end{proof}

\begin{cor}\label{cor-random2}
If $(a)$ correlation coefficient $\rho$ is sufficiently large such that  observations $x_k, \forall k$ fall in two
consecutive intervals, or if $(b)$  $P_t=p_0$ is sufficiently small such that $\tau_{2}^{d}<0$,   then we have $0<f^*_{ii}<1$.
\end{cor}
\begin{proof}
See Appendix \ref{proof-of-corollary2}.
\end{proof}

%
\section{Addressing Problem $(\mathcal{S})$}\label{min p_t cons P_ND}

Let start with problem $(\mathcal{S})$ in (\ref{min_det}) for pure censoring scheme with the optimization variables $\tau_1, \tau_2, t$. Similar to problem $(\mathcal{O})$ in (\ref{min_det}) for pure censoring scheme, 
%
this  problem is not necessarily convex and local minima may be found. Let $\tau_1^d,\tau_2^d, t^d$ denote the solutions to problem $(\mathcal{S})$ in (\ref{min_det}) for pure censoring scheme. Note that, in general this set of solutions is different from the set corresponding to problem $(\mathcal{O})$ in (\ref{min_det}).
Our contribution in this section is addressing problem $(\mathcal{S})$ in (\ref{min_det}) for our proposed CRT schemes.
%
%
In Section \ref{derive PF_PM} we derived $P_M, P_F$ expressions for CRT schemes in terms of $\tau_1,\tau_2, g, f,t$. Assuming sensors partition their observation spaces into the same intervals regardless of the employed scheme, 
in this section we let $\tau_1=\tau_1^d,\tau_2=\tau_2^d$  and view $P_M,P_F$ expressions for CRT schemes as functions of $g,f,t$ only.

Consider problem $(\mathcal{S})$ in (\ref{min_det}) for CRT schemes with the optimization variables $g,f,t$ where $0 \leq g,f \leq 1$. Since the first term of $P_t$ expression in (\ref{definition-P-t}) does not depend on the optimization parameters, we consider the following equivalent problem
\begin{align}
&\min_{g,f,t}~~g P(x_k \in {\cal R}^{0}|\mathcal{H}_0)+ f P(x_k \in {\cal R}^{-1}|\mathcal{H}_0)~~(\mathcal{S'})\nonumber\\
&\text{s.t.}~~P_{M}(g,f,t)\leq\alpha,~~P_F(g,f,t)\leq\beta,~~0 \leq g \leq 1,~~0 \leq f \leq 1\label{S-prime}
\end{align}
For the rest of this section, suppose $g_i^{opt},f_i^{opt},t_i^{opt}$ are the solutions to $(\mathcal{S'})$ for CRT-I scheme and $g_{ii}^{opt},f_{ii}^{opt},t_{ii}^{opt}$ are the solutions to $(\mathcal{S'})$ for CRT-II scheme, where $t_{i}^{opt},t_{ii}^{opt}$ can be different from $t^d$ (i.e., the solution to problem $(\mathcal{S})$ in (\ref{min_det}) for pure censoring scheme).

\subsection{CRT-I Scheme}\label{solve-S-random-i}

Let start with  $P_M,P_F$ expressions in (\ref{definition_P_ND_opt_i}). Recall $P_M,P_F$ expressions are not polynomial functions of $g,f$. Following the same reasoning as in Section \ref{approach-i-solve-O}, we consider instead  $P_M', P_F'$ expressions in (\ref{P'_M,P'_F}), which are polynomial functions of $g,f$. 
%
%
Now, let $(\mathcal{S}'')$ be similar to $(\mathcal{S}')$ in (\ref{S-prime}), with the difference that $P_M,P_F$ are replaced with $P'_M,P'_F$.
To solve $(\mathcal{S''})$, first we decompose it into two subproblems $(\mathcal{S}''_1)$ and $(\mathcal{S}''_2)$ as the following
\begin{align}\label{S1-S2-subproblems-randomized}
&\mbox{given }t~~\min_{g,f}~~g P(x_k \in {\cal R}^{0}|\mathcal{H}_0)+ f P(x_k \in {\cal R}^{-1}|\mathcal{H}_0)~~(\mathcal{S}''_1)\nonumber\\
&\text{s.t.}~~P'_{M}(g,f,t)\leq\alpha,~~P'_F(g,f,t)\leq\beta,~~0 \leq g \leq 1,~~0 \leq f \leq 1\\
&\mbox{ if }P'_{M}=\alpha ~\&~P'_F< \beta
~~\mbox{given }{g},{f}~~\min_{t}~~P'_{M}(g,f,t) ~~(\mathcal{S}''_2)\nonumber\\
&\mbox{ if }P'_{F}=\beta ~\&~P'_M< \alpha
~~\mbox{given }{g},{f}~~\min_{t}~~P'_{F}(g,f,t) ~~(\mathcal{S}''_2)\nonumber
\end{align}
%
%
%
Suppose $g_i^*,f_i^*,t_i^*$ are the final solutions after solving $(\mathcal{S}''_1)$ and $(\mathcal{S}''_2)$. We find $g_i^*,f_i^*,t_i^*$ as the following.

$\bullet$ {\textbf{Solving $(\mathcal{S}''_1)$:}} 
We recognize that  $(\mathcal{S}''_1)$  is an extension of geometric programming (GP) problems (so-called signomial programming \cite{c24}), since the constraints on  $P_{M}'(g,f,t), P_{F}'(g,f,t)$ can be decomposed as the following,  where the terms generated from the decomposition $P_{M}'^{1}(g,f,t)$, $P_{M}'^{2}(g,f,t)$, $P_{F}'^{1}(g,f,t)$, $P_{F}'^{2}(g,f,t)$ are all posynomials. 
\begin{align}
&P_{M}'(g,f,t)=
\underbrace{\sum_{n=0}^{K}\sum_{m=0}^{K}\gamma^{1}_{n,m}f^ng^m}_{=P_{M}'^{1}(g,f,t)}  -\underbrace{\sum_{n=0}^{K}\sum_{m=0}^{K}\gamma^{2}_{n,m}f^ng^m}_{=P_{M}'^{2}(g,f,t)} \leq\alpha,\nonumber\\
&P_F'(g,f,t)=
\underbrace{\sum_{n=0}^{K}\sum_{m=0}^{K}\delta^{1}_{n,m}f^ng^m}_{=P_{F}'^{1}(g,f,t)}-\underbrace{\sum_{n=0}^{K}\sum_{m=0}^{K}\delta^{2}_{n,m}f^ng^m}_{=P_{F}'^{2}(g,f,t)}  \leq\beta, \nonumber
\nonumber
\end{align}
where $\gamma^{1}_{n,m}, \gamma^{2}_{n,m}, \delta^{1}_{n,m}, \delta^{2}_{n,m}$ are positive functions of $\tau_1,\tau_2$.
Using the above decompositions, the constraints on  $P_{M}'(g,f,t), P_{F}'(g,f,t)$ can be expressed as
%
\begin{equation}\label{ratio-constraint}
\frac{P_{M}'^{1}(g,f,t)}{1+ \frac{P_{M}'^{2}(g,f,t)}{\alpha}} \leq \alpha, ~~~\frac{P_{F}'^{1}(g,f,t)}{1+ \frac{P_{F}'^{2}(g,f,t)}{\beta}} \leq \beta
\end{equation}
The problem $(\mathcal{S}''_1)$ in hand still cannot be turned into a convex problem. However, we find the GP approximation of this problem (which we refer to as $(\mathcal{S}^{''gp}_{1})$ in (\ref{S-ini})), by approximating each ratio in (\ref{ratio-constraint}) with a posynomial \cite{c24}. To accomplish this, we approximate each denominator in (\ref{ratio-constraint}) with a monomial (using the arithmetic-geometric mean inequality) and leave the numerators unchanged. While the ratio of two posynomials is not a posynomial, the ratio  between a posynomial and a monomial is another posynomial. 
Let  $g',f'$ be a feasible point in $(\mathcal{S}''_1)$. Recall the arithmetic-geometric mean inequality states $\sum_n\sum_m \mu_{n,m}\geq \prod_n\prod_m (\frac{\mu_{n,m}}{\nu_{n,m}})^{\nu_{n,m}}$ where $\sum_n\sum_m \nu_{n,m}=1$ \cite{c24}. Using this inequality 
we find (\ref{right side constraint rewrite_P_M}), (\ref{right side constraint rewrite_P_F}) given below. 
\begin{figure*}[!htbp]
\begin{align}
&1+\frac{P_{M}'^{2}(g,f,t)}{\alpha}=\sum_{n=0}^{K}\sum_{m=0}^{K}\frac{\gamma^{2}_{n,m}}{\alpha}f^ng^m {\geq} \prod_{n=0}^{K} \prod_{m=0}^{K} (\frac{P_{M}'^{2}(g',f',t)}{\frac{\gamma^{2}_{n,m}}{\alpha}(f')^n(g')^m}\frac{\gamma^{2}_{n,m}}{\alpha}f^ng^m)
^{\frac{\frac{\gamma^{2}_{n,m}}{\alpha}(f')^n(g')^m}{P_{M}'^{2}(g',f',t)}}= \tilde{P}_M (g,f,g',f',t) \label{right side constraint rewrite_P_M}\\
&1+\frac{P_{F}'^{2}(g,f,t)}{\beta}=\sum_{n=0}^{K}\sum_{m=0}^{K}\frac{\delta^{2}_{n,m}}{\beta}f^ng^m {\geq} \prod_{n=0}^{K} \prod_{m=0}^{K} (\frac{P_{F}'^{2}(g',f',t)}{\frac{\delta^{2}_{n,m}}{\beta}(f')^n(g')^m}\frac{\delta^{2}_{n,m}}{\beta}f^ng^m)^{\frac{\frac{\delta^{2}_{n,m}}{\beta}(f')^n(g')^m}{P_{F}'^{2}(g',f',t)}}=\tilde{P}_F (g,f, g',f', t)\label{right side constraint rewrite_P_F}
\end{align}
\hrulefill
\end{figure*}
\begin{figure*}[!htbp]
\scriptsize
\begin{align}\label{feasibility_S_1}
\frac{\alpha P_{M}'^{1}(g_{i}^{*gp}, f_{i}^{*gp}, t)}{\alpha+P_{M}'^{2}(g_{i}^{*gp},f_{i}^{*gp},t)}\overset{(a)}{\leq}
\frac{P_{M}'^{1}(g_{i}^{*gp},f_{i}^{*gp},t )}{\tilde{P}_{M}(g_{i}^{*gp},f_{i}^{*gp}, g',f', t)}\overset{(b)}{\Rightarrow}
\frac{\alpha P_{M}'^{1}(g_{i}^{*gp}, f_{i}^{*gp}, t)}{\alpha+P_{M}'^{2}(g_{i}^{*gp},f_{i}^{*gp},t)}{\leq}\alpha
\Rightarrow P_{M}'^{1}(g_{i}^{*gp},f_{i}^{*gp},t)-P_{M}'^{2}(g_{i}^{*gp},f_{i}^{*gp},t)\leq\alpha ~ \overset{(c)}{\Rightarrow}P_{M}'(g_{i}^{*gp},f_{i}^{*gp},t)\leq \alpha
\end{align}
\normalsize
\hrulefill
\end{figure*}
%
%
%
Using  (\ref{right side constraint rewrite_P_M}), (\ref{right side constraint rewrite_P_F}) to replace the constraints in (\ref{ratio-constraint}),  we form the following problem, that is the GP approximation of $(\mathcal{S}''_1)$ and its feasible region contains that of $(\mathcal{S}''_1)$.
%
\begin{align}\label{S-ini}
&\mbox{given }g',f',t~~\min_{g,f}~~g P(x_k \in {\cal R}^{0}|\mathcal{H}_0)+ f P(x_k \in {\cal R}^{-1}|\mathcal{H}_0)~(\mathcal{S}^{''gp}_{1})\nonumber\\
&\text{s.t.}~~\frac{P_{M}'^{1}(g,f,t)}{\tilde{P}_{M}(g,f,g',f',t)}\leq \alpha,~~\frac{P_{F}'^{1}(g,f,t)}{\tilde{P}_{F}(g,f,g',f',t)}\leq \beta, \nonumber\\
& ~~0 \leq g \leq 1,~~0 \leq f \leq 1
\end{align}
%
Note that $(\mathcal{S}^{''gp}_{1})$ is GP and we can carry out an iterative procedure to solve  it numerically until it converges to a solution (i.e., the difference between the computed optimizers in two consecutive iterations becomes smaller than a pre-determined threshold). Suppose $g_{i}^{*gp},f_{i}^{*gp}$ are the solutions to $(\mathcal{S}^{''gp}_{1})$. 
We can establish (\ref{feasibility_S_1})
where $(a)$ follows from (\ref{right side constraint rewrite_P_M}), $(b)$ is obtained from the first constraint in $(\mathcal{S}^{''gp}_{1})$ and $(c)$ follows from the equality $P_{M}'=P_{M}'^{1}-P_{M}'^{2}$. We can similarly show $P_{F}'(g_{i}^{*gp},f_{i}^{*gp}, t) \leq \beta$. Using this inequality and the last inequality $P_{M}'(g_{i}^{*gp},f_{i}^{*gp},t)\leq \alpha$ in (\ref{feasibility_S_1}) one can easily verify that $g_{i}^{*gp},f_{i}^{*gp}$ is a feasible point in $(\mathcal{S}''_1)$. The solution $g_{i}^{*gp},f_{i}^{*gp}$ to which we converge depends on the very first chosen feasible point $g',f'$ and hence it is important to find a good starting point  $g'=g_{i_0}^{gp}, f'=f_{i_0}^{gp}$ when solving $(\mathcal{S}^{''gp}_{1})$.
A robust strategy to obtain a good starting point is to form another GP approximation of $(\mathcal{S}''_1)$, which we call $(\mathcal{S}^{ini})$, via approximating the constraints in (\ref{S1-S2-subproblems-randomized}) and replacing 
the terms $1-f$ and $1-g$ in $P'_M, P'_F$ expressions of (\ref{P'_M,P'_F})  with $\frac{1}{4f}$ and $\frac{1}{4g}$, respectively. Let $P''_M, P''_F$ denote the new expressions after these replacements. Since $1-x\leq\frac{1}{4x}$ for $x \in \mathbbm{R}^+$ we have 
$P_{M}''(g,f,t) \leq P_{M}'(g,f,t)\leq\alpha$ and $P_{F}''(g,f,t) \leq P_{F}'(g,f,t)\leq\beta$,
and therefore every feasible point in $(\mathcal{S}^{ini})$ is also a feasible point in $(\mathcal{S}''_1)$. With the new constraints $(\mathcal{S}^{ini})$ is GP and we can take an iterative approach to solve  $(\mathcal{S}^{ini})$ numerically until it converges to a solution $g_{i_0}^{gp},f_{i_0}^{gp}$. We let this solution be the very first starting point  $g',f'$ for solving $(\mathcal{S}^{''gp}_{1})$.
%
%
%
In summary, to tackle $(\mathcal{S}''_1)$ in (\ref{S1-S2-subproblems-randomized}), we find two GP approximations of $(\mathcal{S}''_1)$, namely $(\mathcal{S}^{''gp}_{1})$, $(\mathcal{S}^{ini})$. 
Solving $(\mathcal{S}^{ini})$ first provides us with a very good starting point for $(\mathcal{S}^{''gp}_{1})$. With the good starting point,  we solve $(\mathcal{S}^{''gp}_{1})$ to find $g_{i}^{*gp},f_{i}^{*gp}$. Having $g_{i}^{*gp},f_{i}^{*gp}$ we can now proceed to solve $(\mathcal{S}''_2)$ in (\ref{S1-S2-subproblems-randomized}).

$\bullet$ {\textbf{Solving $(\mathcal{S}''_2)$:}} With the solution obtained from solving $(\mathcal{S}''_1)$, we check whether $P'_{M}(g,f,t)=\alpha$, $P'_{F}(g,f,t)<\beta$, or $P'_{M}(g,f,t)<\alpha$, $P'_{F}(g,f,t)=\beta$. Similar to the method we conduct to solve $(\mathcal{O}''_2)$ in Section \ref{solve-O-CRT-I}, we adjust $t$ until the inequality constraint holds with equality  in the former case $P'_{F}(g,f,t)=\beta$, or  in the latter case $P'_{M}(g,f,t)=\alpha$.
We carry out an iterative procedure to iterate between solving $(\mathcal{S}^{ini}),(\mathcal{S}^{''gp}_{1})$  and solving $(\mathcal{S}''_2)$ until convergence is reached. Note that at each iteration the solution of $(\mathcal{S}^{''gp}_{1})$ is still a feasible point in $(\mathcal{S}''_1)$.
We refer to $g_i^*,f_i^*,t_i^*$ as the solutions corresponding to the convergence. 
%
%
Although in general $g_i^{opt} \neq g_i^*, f_i^{opt} \neq f_i^*,t_i^{opt} \neq t_i^*$, using a similar argument to Lemma \ref{lemma-one} of Section \ref{approach-i-solve-O}, we can show that if $0 \leq f^{*}_i<1$ and $0<g^{*}_i \leq 1$ then $ 0< f_i^{opt}<1$ and $0< g_i^{opt}< 1$. This result implies that CRT-I scheme is more effective than pure censoring scheme, i.e., it provides a lower transmission probability, under the same constraints on miss detection and false alarm probabilities.  
In Appendix \ref{appendix-proof-cor-random1-S} we show that when the same condition as in Corollary \ref{cor-random1} of Section \ref{approach-i-solve-O} holds we have $0 \leq f^{*}_i<1$ and $0<g^{*}_i \leq 1$. 

%
%
%

\subsection{CRT-II Scheme}\label{solve-S-random-ii}

Let start with $P_M,P_F$ expressions in (\ref{definition P_ND_opt_ii}). Different from Section \ref{solve-S-random-i}, $P_M,P_F$ are polynomials of $g,f$.
Consider $(\mathcal{S}''_1), (\mathcal{S}''_2)$ in (\ref{S1-S2-subproblems-randomized}), with $P_M,P_F$ in (\ref{definition P_ND_opt_ii}). Suppose $g^*_{ii}, f^*_{ii}, t^{*}_{ii}$ are the final solutions after solving $(\mathcal{S}''_1), (\mathcal{S}''_2)$. We find $g^*_{ii}, f^*_{ii}, t^{*}_{ii}$ using the same approach as we have explained in Section \ref{solve-S-random-i}, that is, 
we carry out an iterative procedure to iterate between solving $(\mathcal{S}^{ini}),(\mathcal{S}^{''gp}_{1})$  and solving $(\mathcal{S}''_2)$ until convergence is reached.
%
Although in general $g_{ii}^{opt} \neq g_{ii}^*, f_{ii}^{opt} \neq f_{ii}^*, t_{ii}^{opt} \neq t_{ii}^*$, using a similar argument to Lemma \ref{lemma-one} of Section \ref{approach-i-solve-O} we can also show that if $0 \leq f^{*}_{ii}<1$ and $0<g^{*}_{ii} \leq 1$ then $ 0< f^{opt}_{ii}<1$ and $0< g^{opt}_{ii}< 1$. This result implies that CRT-II scheme is more effective than pure censoring scheme, i.e., it provides a lower transmission probability, under the same constraints on miss detection and false alarm probabilities.  
Following a similar argument to Appendix  \ref{appendix-proof-cor-random1-S} we can show that when the same conditions as in  Corollary \ref{cor-random2} of Section \ref{solve-O-CRT-II} hold we have $0 \leq f^{*}_{ii}<1$ and $0<g^{*}_{ii} \leq 1$. 


\section{Numerical Results}\label{numerical results}
In this section, through Matlab simulations, we corroborate  
our analytical results in  sections \ref{min P_ND cons P_F} and \ref{min p_t cons P_ND} for solving problems $(\mathcal{O})$ and $(\mathcal{S})$ in (\ref{min_det}) for CRT-I and CRT-II schemes, and compare the performances of CRT-I and CRT-II schemes against that of pure censoring scheme. 
Considering our signal model in Section \ref{System_model}, we let $K\!=\!5$, $\sigma^2_{v}=-50$dBm, ${\cal A}=1$, and define communication SNR and sensing SNR (both in dB), denoted as SNR$_h$ and  SNR$_c$, respectively, where SNR$_h \!=\! 10 \log_{10} ( \frac{\sigma_{h}^2}{\sigma_v^2})$ and 
SNR$_c\!=\! 10 \log_{10} ( \frac{{\cal A }^2}{ \sigma_w^2})$.  In our simulations we vary SNR$_h$ and  SNR$_c$ by changing $\sigma_{h}^2$ and $\sigma_w^2$. 
We compare $P_M$ values achieved by CRT-I and CRT-II schemes against that of pure censoring scheme, as we vary different variables in our problem setup, including SNR$_h$, SNR$_c$, $p_0$ (maximum transmission probability), $\beta$ (largest tolerable $P_F$), and correlation coefficient $\rho$, and
investigate the conditions under which CRT-I and CRT-II schemes outperform pure censoring scheme.

%
\subsection{Performance Comparison when Solving Problem $(\mathcal{O})$}\label{numerical-O}
%
We start with solving $(\mathcal{O})$  in (\ref{min_det}) for pure censoring and obtain the local thresholds $\tau_1^d, \tau_2^d$ as well as the FC threshold $t$. The $P_M,P_F$ expressions for pure censoring scheme are found from (\ref{definition_P_ND_opt_i}) or (\ref{definition P_ND_opt_ii}) by letting $g=0,f=1$.
%
As mentioned in Section \ref{min P_ND cons P_F}, we use these obtained local thresholds when solving  $(\mathcal{O})$ for CRT schemes. 
%
%
To evaluate the performance of our proposed CRT-I scheme we consider two scenarios, which we refer to as ``CRT-I" and ``CRT-I with $f=1$ at FC" here. The values reported in the tables and figures  for  ``CRT-I with $f=1$ at FC" are based on our analytical results in Section \ref{solve-O-CRT-I}, where we solve $(\mathcal{O}''_1), (\mathcal{O}''_2)$ using $P_M',P_F'$  in (\ref{P'_M,P'_F}) and find $f^*,g^*,t^*$ as described in Section \ref{solve-O-CRT-I},  and then evaluate $P_M,P_F$ in (\ref{definition_P_ND_opt_i}) at $f^*,g^*,t^*$. On the other hand, the values reported for ``CRT-I" are obtained from solving sub-problems $(\mathcal{O}'_1), (\mathcal{O}'_2)$ in (\ref{define-O-prime}) using $P_M,P_F$  in (\ref{definition_P_ND_opt_i}). In the absence of analytical solution to sub-problems $(\mathcal{O}'_1), (\mathcal{O}'_2)$ in (\ref{define-O-prime}), we 
solve these sub-problems and find $f^*,g^*,t^*$ through numerical search, and then calculate the corresponding $P_M,P_F$ values. 
To evaluate the performance of our proposed CRT-II scheme we use our analytical results in Section \ref{solve-O-CRT-II}, where we solve $(\mathcal{O}'_1), (\mathcal{O}'_2)$ in (\ref{define-O-prime}) using $P_M,P_F$ in (\ref{definition P_ND_opt_ii}) and find $f^*,g^*,t^*$ as described in Section \ref{solve-O-CRT-II}, and then evaluate $P_M,P_F$ in  (\ref{definition P_ND_opt_ii}) at $f^*,g^*,t^*$. Due to space limitations, the values of $t^*$ are not listed in the tables.

%
%
%
$\bullet$ {\bf Performance Comparison when $p_0$ Varies}:
Table \ref{table_problem_O_rho_07_w/channel} on the left compares the performances of pure censoring, CRT-II, CRT-I with $f=1$ at FC, and CRT-I,  for SNR$_h=5$dB, SNR$_c=10$dB, $\beta=0.01$, $\rho=0.5$, when $p_0$ takes values of $p_0=0.4, 0.6, 0.8$.
Table \ref{table_problem_O_rho_07_w/channel} on the right compares the same, with the difference that $\rho=0.7$. 
Clearly, $P_M$ values in the tables indicate  that $P_M^{\tiny{\mbox{CRT-II}}}<P_M^{\tiny{\mbox{CRT-I}}} < P_M^{\tiny{\mbox{CRT-I~with~f=1~at~FC}}} < P_M^{\tiny{\mbox{pure~censoring}}}$ (exception is $p_0=0.8, \rho=0.5$), that is, CRT schemes outperform pure censoring scheme. This confirms that randomized transmission can improve detection performance under communication rate constraint, when sensors' observations conditioned on each hypothesis are dependent. As expected, performance of CRT-II (in which the FC makes use of information about sensor decision rules and knowledge of realizations $\{r_{g_k}, r_{f_k}\}_{k=1}^K$) is better than CRT-I and CRT-I with $f=1$ at FC (in which the FC does not have this information and does not know these realizations). Also, CRT-I outperforms CRT-I with $f=1$ at FC, since both CRT-I and CRT-II use the knowledge of $g,f$ in constructing the fusion rule, whereas CRT-I with $f=1$ at FC ignores the knowledge of $g,f$ in constructing the fusion rule. 
%
%
Examining $f^*$ values  we observe that $f^*$ increases and approaches to one (i.e., pure censoring scheme without randomized transmission) as  $p_0$ increases.
%
%
This observation can be explained as the following. 
%
As $p_0$ increases, the thresholds $\tau_1^d,\tau_2^d$ become closer to each other, such that the length of censoring interval $\mathcal{R}^0=[\tau_2^d,\tau_1^d]$ decreases. 
%
%
Consequently, the chances that all observations $x_k$'s fall in two consecutive intervals reduce, i.e., the chances that the condition in Corollary \ref{cor-random1} for CRT-I or the conditions in Corollary \ref{cor-random2} for CRT-II are satisfied decrease, and $f^*$ approaches one.
We also note that $f^{*}$ values for CRT-I with $f=1$ at FC is close to one. Particularly, when $p_0=0.8,\rho=0.5$ we have
%
%
$P_M^{\tiny{\mbox{CRT-I~with~f=1~at~FC}}}  \approx P_M^{\tiny{\mbox{pure~censoring}}}$. As correlation increases from $\rho=0.5$  to 
$\rho=0.7$ the value of  $f^{*}$ for CRT-I with $f=1$ at FC reduces and differs from one, and CRT-I with $f=1$ at FC starts to outperform pure censoring. 
%
%
%
%
%
Table \ref{table_problem_O_rho_05_wo/channel} compares the performances of pure censoring, CRT-II, CRT-I with $f=1$ at FC, and CRT-I,  for SNR$_h=10$dB, SNR$_c=10$dB, $\beta=0.01$, $\rho=0.5$, when $p_0$ takes values of $p_0=0.4, 0.6, 0.8$. Note that the simulation parameters are similar to those of Table \ref{table_problem_O_rho_07_w/channel},  with the difference that SNR$_h=10$dB. For Table  \ref{table_problem_O_rho_05_wo/channel}  we can make observations similar to those we made for Table \ref{table_problem_O_rho_07_w/channel}.

%
$\bullet$ {\bf  Performance Comparison  when $\rho$ Varies}: 
Table \ref{table_problem_O_rho_varies_w/channel} on the left compares the performances of pure censoring, CRT-II, CRT-I with $f=1$ at FC, and CRT-I,  for SNR$_h=5$dB, SNR$_c=10$dB, $\beta=0.01$, $p_0=0.4$, when $\rho$  takes values of $\rho=0.1, 0.3, 0.5, 0.7, 0.9$. Examining $P_M$ values we note that at low correlation $\rho=0.1$ pure censoring and CRT schemes perform closely. For $\rho > 0.1$, CRT schemes start to outperform  pure censoring, i.e., effect of randomized transmission on improving detection performance becomes more significant as $\rho$ increases. Comparing CRT schemes, Table \ref{table_problem_O_rho_varies_w/channel} on the left  suggests that $P_M^{\tiny{\mbox{CRT-II}}}<P_M^{\tiny{\mbox{CRT-I}}} < P_M^{\tiny{\mbox{CRT-I~with~f=1~at~FC}}}$ for all $\rho > 0.1$. 
%
%
%
Examining  $f^*$ values  we note that 
as $\rho$ increases $f^*$ value for CRT-I (exception  is $\rho=0.9$) and CRT-II decrease, indicating that the detection performance enhancement due to randomized transmission in CRT-I and CRT-II becomes more notable at higher correlation. 
For instance, at $\rho=0.5$, CRT-II and CRT-I improve upon pure censoring by $18$\%  and $12$\%, respectively. 
%
%
Table \ref{table_problem_O_rho_varies_w/channel} on the right considers the special case of $\rho=0$ and compares the performances of pure censoring, CRT-II, CRT-I with $f=1$ at FC, and CRT-I,  for SNR$_h=5$dB, SNR$_c=12$dB, $\beta=0.01$, when $p_0$ takes values of $p_0=0.4, 0.5, 0.8$.
%
%
%
%
For CRT-I scheme, we observe that as $p_0$ increases $f^{*}=1, g^*=0$ ($f^*,g^*$ remain unchanged) and $P_M$ values of  CRT-I and pure censoring schemes are similar. This is in agreement with Corollary 1, which states $f^* \neq 1$ for sufficiently large $\rho$. 
On the other hand, for CRT-II scheme, as $p_0$ increases $f^{*}$ approaches one, and at $p_0=0.8$,  $P_M$ values of  CRT-II and pure censoring schemes are similar.
%
This is consistent with Corollary \ref{cor-random2}, that states $f^* \neq 1$ when either $\rho$ is sufficiently large or $p_0$ is sufficiently small. 
%

$\bullet$ {\bf Effect of Correlation Mismatch}: The data in Table \ref{table_correlation_mismatch} explores the effect of incorrect correlation information (correlation mismatch) on the performance of pure censoring scheme for SNR$_h=5$dB, SNR$_c=10$dB, $\beta=0.01$, $p_0=0.4, 0.6, 0.8$, as the actual correlation $\rho$ value varies.
Correlation mismatch in our problem setup means that the fusion rule at FC ignores the actual correlation information and employs a fusion rule as if the sensors' observations are conditionally independent ($\rho=0$). 
%
%
%
Table \ref{table_correlation_mismatch} shows that, although the first constraint when solving ($\mathcal{O}$) is satisfied and the transmission probability $P_t$ is upper bounded by the given $p_0$ value, the second constraint in the problem (the constraint on false alarm probability $P_F$)  does not hold and all $P_F$ values exceed the largest tolerable $P_F$ (i.e., all $P_F$ values are larger than $\beta=0.01$).
%
%

%
$\bullet$ {\bf  Performance Comparison when SNR$_h$ Varies}: 
Fig. \ref{P_M-versus-channel-fig} shows $P_M$ versus SNR$_h$ when SNR$_c=10$dB, $\beta=0.01$, $p_0=0.4, \rho=0.5 $. 
This figure shows that for SNR$_h \leq 15$dB we have $P_M^{\tiny{\mbox{CRT-II}}}<P_M^{\tiny{\mbox{CRT-I}}} < P_M^{\tiny{\mbox{CRT-I~with~f=1~at~FC}}} < P_M^{\tiny{\mbox{pure~censoring}}}$, that is, CRT schemes outperform pure censoring and CRT-II provides the largest performance gain (with respect to pure censoring). The performance gain due to randomized transmission diminishes as SNR$_h$ exceeds  $15$dB and the performances of CRT schemes converge to that of pure censoring. 
For instance, at SNR$_h=10$dB, CRT-II and CRT-I improve upon pure censoring by $15$\%  and $8$\%, respectively. 
%
%
%

%
$\bullet$ {\bf Performance Comparison when $\beta$ Varies}: 
Fig. \ref{P_M-versus-P_f-fig} plots $P_M$ versus $\beta$ when SNR$_h=10$dB, SNR$_c=10$dB, $p_0=0.4, \rho=0.5 $. 
This figure shows that for $ \beta \leq 0.05$,  CRT schemes outperform pure censoring  and CRT-II provides the largest performance gain. The performance gain due to randomized transmission reduces   and the performances of CRT schemes converge to that of pure censoring scheme for $ \beta > 0.05$.  
For instance, at $\beta=0.05$, CRT-II and CRT-I improve upon pure censoring by $22$\%  and $15$\%, respectively. 
%
%

%
\subsection{Performance Comparison when Solving Problem $(\mathcal{S})$}
Similar to Section \ref{numerical-O}, to evaluate the performance of our proposed CRT-I scheme we consider two scenarios, which we refer to as ``CRT-I" and ``CRT-I with $f=1$ at FC" here. 
$\bullet$ {\bf  Performance Comparison when $\rho$ Varies}: 
Table \ref{table_problem_S_rho_varies_w/channel} on the left compares the performances of pure censoring, CRT-II, CRT-I with $f=1$ at FC, and CRT-I, for SNR$_h=5$dB, SNR$_c=10$dB, $\beta=0.01,\alpha=0.1$, when $\rho$  takes values of $\rho=0.1, 0.3, 0.5, 0.7, 0.9$. Examining $P_t$ values  we note that at low correlation $\rho=0.1$ pure censoring and CRT schemes perform closely. For $\rho > 0.1$, CRT schemes start to outperform  pure censoring, i.e., effect of randomized transmission on improving detection performance becomes more significant as $\rho$ increases. 
Comparing CRT schemes, Table \ref{table_problem_S_rho_varies_w/channel} on the left  suggests that $P_M^{\tiny{\mbox{CRT-II}}}<P_M^{\tiny{\mbox{CRT-I}}} < P_M^{\tiny{\mbox{CRT-I~with~f=1~at~FC}}}$ for all $\rho > 0.1$. 
Examining  $f^*$ values we note that 
as $\rho$ increases $f^*$ value for CRT-I and CRT-II decrease, indicating that the detection performance enhancement due to randomized transmission in CRT-I and CRT-II becomes more notable at higher correlation. 
For instance, at $\rho=0.5$, CRT-II and CRT-I improve upon pure censoring by $24$\%  and $13$\%, respectively. 
%
%
Table \ref{table_problem_S_rho_varies_w/channel} on the right considers the special case of $\rho=0$ and compares the performances of pure censoring, CRT-II, CRT-I with $f=1$ at FC, and CRT-I, for SNR$_h=5$dB, SNR$_c=12$dB, $\beta=0.01$, $\alpha=0.025$.
For CRT-I scheme, we observe that $f^{*}=1, g^*=0$  and $P_t$ values of CRT-I and pure censoring schemes are similar. 
%
%
This is in agreement with Corollary 1, which states $f^* \neq 1, g^*\neq 0$ for sufficiently large $\rho$. 
%
On the other hand, for CRT-II scheme, $f^{*} \neq 1, g^*\neq 0$ and $P_t$ value of CRT-II is smaller than that of pure censoring scheme. 
%
This is consistent with Corollary \ref{cor-random2}, that states $f^* \neq 1, g^*\neq 0$ when either $\rho$ is sufficiently large or $P_t$ is sufficiently small. 
%
%

$\bullet$ {\bf  Performance Comparison when SNR$_h$ Varies}: 
Fig. \ref{p_t-versus-channel-fig} shows $P_t$ versus SNR$_h$ when SNR$_c=10$dB$, \beta=0.01, \alpha=0.06, \rho=0.5$. This figure shows that for SNR$_h \leq 15$dB we have $P_t^{\tiny{\mbox{CRT-II}}}<P_t^{\tiny{\mbox{CRT-I}}} < P_t^{\tiny{\mbox{CRT-I~with~f=1~at~FC}}} < P_t^{\tiny{\mbox{pure~censoring}}}$, that is, CRT schemes outperform pure censoring scheme and CRT-II provides the largest performance gain. The performance gain due to randomized transmission diminishes as SNR$_h$ exceeds $15$dB and the performances of CRTs converge to that of pure censoring. 
For instance, at SNR$_h=8$dB, CRT-II and CRT-I improve upon pure censoring by $30$\%  and $25$\%, respectively. 
%

$\bullet$ {\bf Performance Comparison  when $\beta$ Varies}: 
Fig. \ref{p_t-versus-P_f-fig} plots $P_t$ versus $\beta$ when SNR$_h=10$dB, SNR$_c=10$dB, $\alpha=0.06, \rho=0.5$. 
This figure shows that for $ \beta \leq 0.03$,  CRT schemes outperform pure censoring scheme and CRT-II provides the largest performance gain. The performance gain due to randomized transmission reduces  and the performances of CRT schemes converge to that of pure censoring scheme for $ \beta > 0.03$.  
For instance, at $\beta=0.015$, CRT-II and CRT-I improve upon pure censoring by $19$\%  and $15$\%, respectively. 
\section{Conclusions}

Considering a binary distributed detection problem, where the FC is tasked with detecting a known signal in correlated Gaussian noises, we proposed two randomized transmission schemes (so-called CRT-I and CRT-II schemes). To investigate the effectiveness of these schemes to improve the system performance, under communication rate constraint,
we formulated and addressed two system-level constrained optimization problems
%
%
and proposed different optimization techniques  to solve these two problems. 
While independent randomization strategy cannot improve detection performance when sensors are restricted to transmit discrete values (over bandwidth constrained error-free channels), for conditionally independent observations
\cite{c21}, our results show that, for conditionally dependent observations,  
our simple and easy-to-implement CRT schemes can improve detection performance, when sensors transmit discrete values over noisy channels.
Through analysis and simulations, we explored and provided the conditions 
under which CRT schemes outperform pure censoring scheme, and illustrated the deteriorating effect of incorrect correlation information (correlation mismatch) on the detection performance.
%
%
When solving the first problem,  our numerical results indicate that CRT schemes outperform pure censoring scheme for SNR$_h \leq $15dB, $ 0.3 < \rho < 0.9$, $\beta \leq 0.05$. 
When solving the second problem,  our numerical results show that  CRT schemes outperform pure censoring scheme for SNR$_h \leq $15dB, $ 0.1 < \rho $, $\beta \leq 0.03$. 
%
%
Also, CRT-II scheme performs better than than CRT-I scheme, for instance, when solving the first (second) problem at $\rho=0.5$, CRT-II and CRT-I improve upon pure censoring by $18$\%  ($35$\%) and $12$\%($17$\%), respectively. 


\appendix


\subsection{Proof of Lemma \ref{lemma-one}}\label{appendix-proof-of-lemma1}
Since $t^d$ is the solution to problem $(\mathcal{O})$ in (\ref{min_det})  for pure censoring scheme it is a feasible point of $(\mathcal{O})$. Also,  $(\mathcal{O})$  is equivalent to problem $(\mathcal{O}')$ when $f=1$. Therefore, the constraints are satisfied $P_F(t^d,1)= P'_F(t^d,1)\leq \beta$, $P_t=p_0$ and $P_{M}(t^d,1)=P'_{M}(t^d,1)$. For the moment, assume $t_i^{opt}\neq t^{d}, 0<f_i^* <1, f_i^{opt}=1.$
Now, consider $(\mathcal{O}''_1)$ given $t=t^d$. For $f=1$ all the constraints are satisfied and therefore $f=1$ is a feasible point of $(\mathcal{O}''_1)$.  Under the assumption $0<f_i^* < 1, f_i^{opt}=1$ and using the above argument we have
\begin{equation}\label{contradict_1}
P'_{M}(t^d,f^{*}_i)<P'_{M}(t^d,1)=P_{M}(t^d,1)
\end{equation}
On the other hand, we have
$P_{M}(t_i^{opt},1)<P_{M}(t^d, 1)$, contradicting the fact that $t^d$ is the solution of problem $(\mathcal{O})$. This implies our assumption above cannot be true and if $f_i^{opt}=1$  then we must have $t_i^{opt}= t^{d}$. 

So, let instead assume $t_i^{opt}=t^d, 0<f_i^* < 1, f_i^{opt}=1$.
Recall $(t^d,f^{*}_i)$ is a feasible point in $(\mathcal{O}')$ since the constraints are satisfied. Hence $P_{M}(t^d,1)<P_{M}(t^d,f^{*}_i)$. On the other hand, we know $P_{M}(t^d,f^{*}_i)<P'_{M}(t^d,f^{*}_i)$. Combining the last two inequalities we reach
$P_{M}(t^d,1)<P'_{M}(t^d,f^{*}_i)$.
However, the latest inequality contradicts (\ref{contradict_1}). This proves our assumptions cannot be true and if $0<f^{*}_i<1$ then $0<f_i^{opt}<1$.


\subsection{Proof of Theorem \ref{dP_M_df_thm_opt_i}}\label{appendix-proof-dP_M_df_thm_opt_i}

 We first consider $P_M'$ in (\ref{P'_M,P'_F}), where  $g$ is a function of $f$. We find that $\frac{d P'_M(t^d,f)}{df}|_{f=1}=\sum_{a_1}\frac{d {\cal M}(t^d,f,a_1)}{df}|_{f=1}$ where ${\cal M}(t^d,f,a_1)$ is given in (\ref{dP_M_df_L_R_opt_ii_reduced}).
\begin{figure*}[!htbp]
\begin{align}
&{\cal M}(t^d,f,a_1)=\frac{1}{a_1!}\sum_{a_2,a_3,a_4,a_5}
\frac{K!(1-P_{u}(1,0, t^d,\boldsymbol{a})) P_{x1}(\boldsymbol{a})}{ a_2! a_3! a_4! a_5!}
(1-g(f))^{a_2} g(f)^{a_3} (1-f)^{a_4} f^{a_5}\label{dP_M_df_L_R_opt_ii_reduced}\\
&\frac{d {\cal M}(t^d,f,a_1)}{df}|_{f=1}
=\sum_{a_5}\frac{K!(P_{u}(1,0,t^d,\boldsymbol{a}^L_{10})-P_{u}(1,0,t^d,\boldsymbol{a}^L_{01}))}{a_1! (K-a_1-a_5-1)!a_5!}P_{x1}(\boldsymbol{a}^L_{00})(-\gamma^L_M+(\frac{-dg}{df}))\label{dP_ND_df_L_d}\\
&
\frac{d {\cal M'}(t^d,f,a_1)}{df}|_{f=1}=\sum_{a_5}\frac{K!(P_{u}(1,0,t^d,\boldsymbol{a}^L_{01})-P_{u}(1,0,t^d,\boldsymbol{a}^L_{10}))}{a_1! (K-a_1-a_5-1)!a_5!}P_{x0}(\boldsymbol{a}^L_{00})(-\gamma_F^L+(\frac{-dg}{df}))\label{dPF-df}
\end{align}
\end{figure*}
To express $\frac{d {\cal M}(t^d,f,a_1)}{df}|_{f=1}$ we use the definition of vector 
$\boldsymbol{a}=[a_1,\ldots,a_5]$ in Section \ref{derive PF_PM_opt_i} to define the vectors
$\boldsymbol{a}^L_{01}\!=\![a_1,K{-a_1}-a_5-1,0,1,a_5]$, $\boldsymbol{a}^L_{00}\!=\![a_1,K{-a_1}-a_5,0,0,a_5]$, $\boldsymbol{a}^L_{10}\!=\![a_1,K{ -a_1}-a_5-1,1,0,a_5]$. 
Taking the derivative $\frac{d {\cal M}(t^d,f,a_1)}{df}$ and noting that for $f=1$ all the terms containing $(1-f)^{a_4},~a_4>0$ or $g(f)^{a_3},~a_3>0$ are zero,
the facts that
%
%
$P_{u}(1,0,t^d,\boldsymbol{a}^L_{01})\!=\!P_{u}(1,0,t^d,\boldsymbol{a}^L_{00})$ and $P_{x1}(\boldsymbol{a}^L_{10})\!=\!P_{x1}(\boldsymbol{a}^L_{01})\!=\!P_{x1}(\boldsymbol{a}^L_{00})$, after some algebraic simplifications we obtain (\ref{dP_ND_df_L_d}),
where $\gamma_M^L$ in (\ref{dP_ND_df_L_d}) is defined as below
\begin{align}\label{definition_gamma_L_opt_i}
\gamma_M^L=\frac{\sum_{a_5}\frac{K!(P_{u}(1,0,t^d,\boldsymbol{a}^L_{10})-P_{u}(1,0,t^d,\boldsymbol{a}^L_{01}))}{a_1! (K-a_1-a_5-1)! a_5!}P_{x1}(\boldsymbol{a}^L_{01})}
{\sum_{a_5}\frac{K!(P_{u}(1,0,t^d,\boldsymbol{a}^L_{10})-P_{u}(1,0,t^d,\boldsymbol{a}^L_{01}))}{a_1! (K-a_1-a_5-1)! a_5!}P_{x1}(\boldsymbol{a}^L_{00})}
\end{align}
Now, suppose $\rho$ is sufficiently large such that  $x_k, \forall k$ fall in two consecutive intervals. 
Considering the definition of $\boldsymbol{a}^L_{01} $ we realize that $x_k \in {\cal R}^{-1} \cup {\cal R}^0$,  $x_k \notin {\cal R}^{1}$, implying that $a_1=0$.  
%
Therefore, $P_{x1}(\boldsymbol{a}^L_{00}) \approx 0$ and $\frac{d {\cal M}(t^d,f,a_1)}{df}|_{f=1} \approx 0$ thus $\frac{dP'_M(t^d,f)}{df}|_{f=1} \approx 0$.
%

Next, we consider $ P_F'$ in (\ref{P'_M,P'_F}).  Taking similar steps as above, we find that $\frac{d P'_F(t^d,f)}{df}|_{f=1}=\sum_{a_1}\frac{d {\cal M'}(t^d,f,a_1)}{df}|_{f=1}$ where $\frac{d {\cal M'}(t^d,f,a_1)}{df}|_{f=1}$ is given in (\ref{dPF-df}) and
$\gamma_F^L$ in (\ref{dPF-df}) is obtained by replacing $P_{x1}(\boldsymbol{a}^L_{01})$, $P_{x1}(\boldsymbol{a}^L_{00})$ in (\ref{definition_gamma_L_opt_i}) with $P_{x0}(\boldsymbol{a}^L_{01})$, $P_{x0}(\boldsymbol{a}^L_{00})$,  respectively.
For $\rho\!=\!0$ 
we have $\frac{P_{x0}(\boldsymbol{a}^L_{01})}{P_{x0}(\boldsymbol{a}^L_{00})}=\frac{P(x_k \in {\cal R}^{-1}|\mathcal{H}_0)}{P(x_k \in {\cal R}^{0}|\mathcal{H}_0)}<\frac{P(x_k \in {\cal R}^{-1}\bigcup{\cal R}^{1}|\mathcal{H}_0)}{P(x_k \in {\cal R}^{0}|\mathcal{H}_0)}=\frac{p_0}{1-p_0}$. As $\rho$ increases,  the ratio $\frac{P_{x0}(\boldsymbol{a}^L_{01})}{P_{x0}(\boldsymbol{a}^L_{00})}$ decreases and thus  $\frac{P_{x0}(\boldsymbol{a}^L_{01})}{P_{x0}(\boldsymbol{a}^L_{00})}<\frac{p_0}{1-p_0}$ for every $\rho$. From the definition of $\gamma_F^L$ and using mediant inequality we conclude that $\gamma_F^L<\frac{p_0}{1-p_0}=\frac{-dg}{df}$.
Combining this with the fact that $P_{u}(1,0,t^d,\boldsymbol{a}^L_{01})>P_{u}(1,0,t^d,\boldsymbol{a}^L_{10})$ we find $\frac{d {\cal M'}(t^d,f,a_1)}{df}|_{f=1}>0$ and consequently $\frac{dP'_F(t^d,f)}{df}|_{f=1}>0$ for every $\rho$. This completes our proof of Theorem 1.


\subsection{Proof of Theorem \ref{dP_F_df_thm_opt_ii}}\label{appendix-proof-dP_F_df_thm_opt_ii}
We consider $P_M, P_F$ in (\ref{definition P_ND_opt_ii}), where $g$ is a function of $f$. Since the local thresholds $\tau_1,\tau_2$ are fixed at $\tau_1^d,\tau_2^d$, we first simplify the notations by dropping them from the terms $P_{u_m}(\tau_1,\tau_2, t, \boldsymbol{a})$ and $P_{xm} (\tau_1,\tau_2,  \boldsymbol{a})$ and denoting these probabilities as $P_{u_m}(t,\boldsymbol{a})$ and $P_{xm} (\boldsymbol{a})$, respectively. We obtain $\frac{dP_M(t^d,f)}{df}|_{f=1}$ given in (\ref{theorem-2-first-equation}).
\begin{figure*}[!htbp]
\hrulefill
\begin{align}
&\frac{dP_M(t^d,f)}{df}|_{f=1}=\sum_{a_{1_1},a_{1_0},...,a_{6_1},a_{6_0}} \frac{K!P_{u_0}(t^d,\boldsymbol{a}) P_{x1}(\boldsymbol{a})}{a_{1_1}!a_{1_0}!...a_{6_1}!a_{6_0}!} \frac{d(M(f,a_{2_1},a_{2_0},\ldots,a_{5_1},a_{5_0}) M'(f,a_{1_1},a_{1_0},a_{6_1},a_{6_0}))}{df}|_{f=1} \label{theorem-2-first-equation}\\
&\mbox{where}~M'(f,a_{1_1},a_{1_0},a_{6_1},a_{6_0})=f^{a_{1_1}+a_{1_0}}(1-f)^{a_{6_1}+a_{6_0}}g(f)^{a_{1_1}+a_{6_1}}(1-g(f))^{a_{1_0}+a_{6_0}} \nonumber \\
& M(f,a_{2_1},a_{2_0},...,a_{5_1},a_{5_0})=f^{a_{2_1}+a_{3_1}+a_{5_1}+a_{5_0}}(1-f)^{a_{2_0}+a_{3_0}+a_{4_1}+a_{4_0}}g(f)^{a_{3_1}+a_{3_0}+ a_{4_1}+a_{5_1}}(1-g(f))^{a_{2_1}+a_{2_0}+a_{4_0}+a_{5_0}} \nonumber \\
&\frac{dP_M(t^d,f)}{df}|_{f=1}=
\sum_{a_{1_0}}\frac{d {\cal M}(t^d,f,a_{1_0})}{df}|_{f=1}+
\sum_{a_{2_1}}\sum_{a_{5_0}}\frac{d {\cal M}'(t^d,f,a_{2_1},a_{5_0})}{df}|_{f=1}  \label{theorem-2-second-equation}\\
&\mbox{where}~{\cal M}(t^d,f,a_{1_0})=\frac{1}{a_{1_0}!}\sum_{a_{2_1}, a_{2_0}, ... , a_{5_1}, a_{5_0} }
\frac{K!}{a_{2_1}! a_{2_0}!...a_{5_1}! a_{5_0}!}P_{u_0}(t^d,\boldsymbol{a})P_{x1}(\boldsymbol{a})
M(f,a_{2_1},a_{2_0}\ldots,a_{5_1},a_{5_0})\nonumber\\
&\mbox{and}~{\cal M}'(t^d,f,a_{2_1},a_{5_0})=\frac{1}{a_{2_1}!a_{5_0}!} \sum_{a_{1_1}, a_{1_0}, a_{6_1}, a_{6_0}} 
\frac{K!}{a_{1_1}!a_{1_0}!a_{6_1}!a_{6_0}!}P_{u_0}(t^d,\boldsymbol{a})P_{x1}(\boldsymbol{a})
M'(f,a_{1_1},a_{1_0},a_{6_1},a_{6_0})\nonumber\\
&\frac{d {\cal M}(t^d,f,a_{1_0})}{df}|_{f=1}=\frac{1}{a_{1_0}!}\sum_{a_{2_1}, a_{5_0}}\frac{K!(P(u_0=0|\boldsymbol{C}=\boldsymbol{C}^{\boldsymbol{a}_{01}})-
P(u_0=0|\boldsymbol{C}=\boldsymbol{C}^{\boldsymbol{a}_1}))}{a_{2_1}!a_{5_0}!}
P_{x1}(\boldsymbol{a}_{01})\nonumber\\
&- \frac{dg}{df}\frac{1}{a_{1_0}!}\sum_{a_{2_1}, a_{5_0}}\frac{K!(P(u_0=0|\boldsymbol{C}=\boldsymbol{C}^{\boldsymbol{a}_{01}})-
	P(u_0=0|\boldsymbol{C}=\boldsymbol{C}^{\boldsymbol{a}_2}))}{a_{2_1}!a_{5_0}!}
P_{x1}(\boldsymbol{a}_{01})\nonumber\\
&+ \frac{1}{a_{1_0}!}\sum_{a_{2_1}, a_{5_0}}\frac{K!(P(u_0=0|\boldsymbol{C}=\boldsymbol{C}^{\boldsymbol{a}_{10}})-
	P(u_0=0|\boldsymbol{C}=\boldsymbol{C}^{\boldsymbol{a}_3}))}{a_{2_1}!a_{5_0}!}
P_{x1}(\boldsymbol{a}_{10}) - \frac{dg}{df} \frac{1}{a_{1_0}!} \sum_{a_{2_1}, a_{5_0}}\frac{K!}{a_{2_1}!a_{5_0}!}
P_{x1}(\boldsymbol{a}_{10})\label{expanded-dMf}\\
&\frac{d {\cal M}'(t^d,f,a_{2_1},a_{5_0})}{df}|_{f=1}=\frac{K!}{(K'-1)!a_{2_1}!\ldots a_{5_0}!} \big( P(u_0=0|\boldsymbol{C}=\boldsymbol{C}^{\boldsymbol{a}'_{00}})P_{x1}(\boldsymbol{a}'_{00})-
P(u_0=0|\boldsymbol{C}=\boldsymbol{C}^{\boldsymbol{a}'_2})P_{x1}(\boldsymbol{a}'_2))\nonumber\\
&+ \frac{dg}{df}\frac{K!\big( P(u_0=0|\boldsymbol{C}=\boldsymbol{C}^{\boldsymbol{a}'_1})P_{x1}(\boldsymbol{a}'_1)-
	P(u_0=0|\boldsymbol{C}=\boldsymbol{C}^{\boldsymbol{a}'_{00}})P_{x1}(\boldsymbol{a}'_{00})\big)}{(K'-1)!a_{2_1}!\ldots a_{5_0}!} \label{dM'-t-f}
\end{align}
\hrulefill
\end{figure*}
The reasoning for such a partitioning is that $x_k \in \mathcal{R}^{1}$ for $k \in {\cal K}_{1_1}\cup {\cal K}_{1_0}\cup {\cal K}_{6_1}\cup {\cal K}_{6_0}$ and $x_k \notin \mathcal{R}^{1}$ for $k \in {\cal K}_{2_1}\cup {\cal K}_{2_0}\cup {\cal K}_{3_1}\cup {\cal K}_{3_0} {\cal K}_{4_1}\cup {\cal K}_{4_0} {\cal K}_{5_1}\cup {\cal K}_{5_0}$ according to the definitions in Section \ref{derive PF_PM_opt_ii}. To continue our derivations,  we use the facts that $M'(f,a_{1_1},a_{1_0},a_{6_1},a_{6_0})|_{f=1}=0$ if at least one of $a_{1_1},a_{6_1},a_{6_0}$ is non-zero and  $M'(f,0,a_{1_0},0,0)|_{f=1}=1$. Similarly, $M(f,a_{2_1},a_{2_0},...,a_{5_1}, a_{5_0})|_{f=1}=0$ if at least one of $a_{2_0},a_{3_1},a_{3_0},a_{4_1},a_{4_0},a_{5_1}$ is non-zero and $M(f,a_{2_1},0,0,0,0,0,0,a_{5_0})|_{f=1}=1$. 
Therefore, we reach (\ref{theorem-2-second-equation}). 
In the following, we argue that  both terms in (\ref{theorem-2-second-equation}) are greater than zero and hence $\frac{dP_M(t^d,f)}{df}|_{f=1}>0$ in (\ref{theorem-2-second-equation}). 

Let consider the first term in (\ref{theorem-2-second-equation}). To express $\frac{d {\cal M}(t^d,f,a_{1_0})}{df}|_{f=1}$ 
we use the definition of vector 
$\boldsymbol{a}=[a_{1_1},a_{1_0},\ldots,a_{6_1},a_{6_0}]$ in Section  \ref{derive PF_PM_opt_ii} to define the vectors 
$\boldsymbol{a}_{00}=[0,l_1,l_2,0,\ldots,0,l_3,0,0]$, that is, $a_{1_0}=l_1,a_{2_1}=l_2, a_{5_0}=l_3$, and the remaining entries are zero, 
 $\boldsymbol{a}_{01}=[0,l_1,l_2+1,0,\ldots,0,l_3,0,0]$,  $\boldsymbol{a}_{10}=[0,l_1,l_2,0,\ldots,0,l_3+1,0,0]$, $\boldsymbol{a}_{1}=[0,l_1,l_2,1,\ldots,0,l_3,0,0]$, $\boldsymbol{a}_{2}=[0,l_1,l_2,0,1,0,\ldots,0,l_3,0,0]$, $\boldsymbol{a}_{3}=[0,l_1,l_2,0,\ldots,0,1,0,l_3,0,0]$ and $\boldsymbol{a}_{4}=[0,l_1,l_2,0,\ldots,0,0,1,l_3,0,0]$.
 We also note that from (\ref{definition P_x_opt_i}) we have $P_{x1}(\boldsymbol{a}_{01})=P_{x1}(\boldsymbol{a}_{1})=P_{x1}(\boldsymbol{a}_{2})$ and $P_{x1}(\boldsymbol{a}_{10})=P_{x1}(\boldsymbol{a}_{3})=P_{x1}(\boldsymbol{a}_{4})$.
Taking the derivative $\frac{d {\cal M}(t^d,f,a_{1_0})}{df}$ and taking into account the terms that become zero for $f=1$ we obtain (\ref{expanded-dMf}). 
Next, we argue that, under the stated conditions in Theorem 2,  the first term in (\ref{theorem-2-second-equation}), which is expanded in (\ref{expanded-dMf}), is greater than zero.
Suppose $P_t=p_0$ is sufficiently small such that ${\tau}^d_2<0$. Considering the definition of $\boldsymbol{a}_{10}$ there exists at least one $x_k \in {\cal R}^{-1}$, also  ${\tau}^d_2<0$ and $P(x \in {\cal R}^{-1}|\mathcal{H}_1)\approx 0$. The combination of these implies that
$P_{x1}(\boldsymbol{a}_{10})\approx 0$. This approximation indicates that the third and forth terms in (\ref{expanded-dMf}) are approximately zero.
On the other hand, one can show that $P(u_0=0|\boldsymbol{C}=\boldsymbol{C}^{\boldsymbol{a}_{01}})-
P(u_0=0|\boldsymbol{C}=\boldsymbol{C}^{\boldsymbol{a}_1}) \approx 0$ and $P(u_0=0|\boldsymbol{C}=\boldsymbol{C}^{\boldsymbol{a}_{01}})-
P(u_0=0|\boldsymbol{C}=\boldsymbol{C}^{\boldsymbol{a}_2}) > 0$. 
The former approximation suggests that the first term in (\ref{expanded-dMf}) is approximately zero and the latter inequality implies that the second term in (\ref{expanded-dMf}) is greater than zero. By combining all these we conclude that the first term in (\ref{theorem-2-second-equation}), which is expanded in (\ref{expanded-dMf}), is greater than zero.

Next, we focus on the second term in (\ref{theorem-2-second-equation}). To express $\frac{d {\cal M}'(t^d,f,a_{2_1},a_{5_0})}{df}|_{f=1}$ we denote $K'=K-(a_{2_1}+a_{2_0}+ a_{3_1}+a_{3_0} + a_{4_1}+a_{4_0} +a_{5_1} + a_{5_0})$, 
and define the vectors  $\boldsymbol{a}'_{00}=[0,K',a_{2_1},\ldots,a_{5_0},0,0]$, $\boldsymbol{a}'_1=[1,K'-1,a_{2_1},\ldots,a_{5_0},0,0]$ and $\boldsymbol{a}'_2=[0,K'-1,a_{2_1},\ldots,a_{5_0},1,0]$.  We also note that from (\ref{definition P_x_opt_i}) we have $P_{x1}(\boldsymbol{a}'_{00})=P_{x1}(\boldsymbol{a}'_1)=P_{x1}(\boldsymbol{a}'_{2})$.
Taking the derivative  $\frac{d {\cal M}'(t^d,f,a_{2_1},a_{5_0})}{df}$ and taking into account the terms that become zero for $f=1$ we obtain (\ref{dM'-t-f}).
By taking similar steps to the ones taken for the first term in  (\ref{theorem-2-second-equation}), 
one can show that, under the stated conditions in Theorem 2, we have $P(u_0=0|\boldsymbol{C}=\boldsymbol{C}^{\boldsymbol{a}'_1})P_{x1}(\boldsymbol{a}'_1) - P(u_0=0|\boldsymbol{C}=\boldsymbol{C}^{\boldsymbol{a}'_{00}})P_{x1}(\boldsymbol{a}'_{00}) \approx 0$ and 
$P(u_0=0|\boldsymbol{C}=\boldsymbol{C}^{\boldsymbol{a}'_{00}})P_{x1}(\boldsymbol{a}'_{00}) - P(u_0=0|\boldsymbol{C}=\boldsymbol{C}^{\boldsymbol{a}'_2})P_{x1}(\boldsymbol{a}'_2) >0$. 
The former approximation implies that the first term in (\ref{dM'-t-f}) is approximately zero and the latter inequality indicates  that the second term in (\ref{dM'-t-f}) is greater than zero. Combining all these suggests that the second term in (\ref{theorem-2-second-equation}), which is expanded in (\ref{dM'-t-f}), is greater than zero.

In summary, we have shown that both terms in (\ref{theorem-2-second-equation}) are greater than zero. Therefore $\frac{dP_M(t^d,f)}{df}|_{f=1}>0$. 
Considering $P_F$ in (\ref{definition P_ND_opt_ii}) we can show in a similar way (with some change of notations) that $\frac{dP_F(t^d,f)}{df}|_{f=1}>0$ under the stated conditions in Theorem 2. Due to lack of space and to avoid repetition, this part is omitted. This completes our proof of Theorem 2.


\subsection{Proof of Corollary \ref{cor-random2}}\label{proof-of-corollary2}

Here, we first prove that under condition $(a)$ we have $0<f^*_{ii}<1$.  Note that at $f=1$, $P_M$ of CRT-I and CRT-II schemes have the same values. The same statement is true for $P_F$ values at $f=1$. Also, note $P_F(t^{*}_{i},f^{*}_{i})=P'_F(t^{*}_{ii},f^{*}_{ii})=\beta$. On the other hand, we know $P_M(t^{*}_{ii},f^{*}_{ii})<P'_M(t^{*}_{i},f^{*}_{i})$ since the amount of information available at the FC for CRT-II is greater than that of CRT-I. Also, from Corollary \ref{cor-random1} we know if condition $(a)$ holds, then $P'_M(t^{*}_i,f^{*}_i)<P'_M(t,1)=P_M(t,1)$. Combining all, we reach $P_M(t^{*}_{ii},f^{*}_{ii})<P_M(t,1)$ implying that $0<f^{*}_{ii}<1$. 

Next, we show that under condition $(b)$ we have $0<f^*_{ii}<1$.
The KKT conditions that need to be solved to find $f^{*}_{ii}$ are similar to the ones in (\ref{KKT-O3-1})-(\ref{KKT-O3-3}), where $P'_M,P'_F$ are replaced with $P_M,P_F$. Therefore, our argument here is similar to the proof of Corollary \ref{cor-random1}. 
For $f=1$ from (\ref{KKT-O3-2}), (\ref{KKT-O3-3}) we have $\mu_1 >0, \mu_2=0, \lambda \geq 0$. Also, from Theorem \ref{dP_F_df_thm_opt_ii} we have $\frac{dP_{M}(t^d,f)}{df}|_{f=1}>0$ and $\frac{dP_{F}(t^d,f)}{df}|_{f=1}>0$ under condition (b).  Now, considering (\ref{KKT-O3-1}) we realize that it cannot be satisfied at $f=1$, 
and hence we have $0<f^{*}_{ii}<1$.

\subsection{Proof of $0 \leq f^{*}_i<1$ and $0<g^{*}_i \leq 1$ under the Condition in Corollary \ref{cor-random1} when Solving Problem $(\mathcal{S}''_1)$ for CRT-I Scheme}
\label{appendix-proof-cor-random1-S}

%
Suppose $p_{\text{min}}$ denote the minimum value that the cost function in problem $(\mathcal{S}''_1)$  can achieve, i.e., $g^{*}_iP(x_k \in {\cal R}^0|{\cal H}_0)+f^{*}_iP(x_k \in {\cal R}^{-1}|{\cal H}_0)=p_{\text{min}}$. From this equality we find $g$ in terms of $f$, that is, $g(f)=\frac{p_{\text{min}}-fP(x_k \in {\cal R}^{-1}|{\cal H}_0)}{P(x_k \in {\cal R}^0|{\cal H}_0)}$. 
Given $t$, let $\mathfrak{L}(f,g,\lambda_1,\lambda_1,\mu_1,\mu_2,\mu_3,\mu_4)$ be the Lagrangian for $(\mathcal{S}''_1)$, where $\lambda_1$, $\lambda_2$, $\mu_1$, $\mu_2$, $\mu_3$, $\mu_4$ are Lagrange multipliers for the constraints $P'_{M}(g,f,t)\leq \alpha$, $P'_{F}(g,f,t)\leq \beta$, $f\leq 1$, $f\geq0$, $g\leq1$, $g\geq 0$, respectively.
The associated KKT conditions are given in (\ref{KKT-O4-1})-(\ref{KKT-O4-5}).
\begin{figure*}[!htbp]
\begin{align}
&\frac{d\mathfrak{L}}{df}=P(x_k \in {\cal R}^{-1}|{\cal H}_0)+\lambda_1\frac{dP'_{M}(g,f,t)}{df}+\lambda_2\frac{dP'_{F}(g,f,t)}{df}+\mu_1-\mu_2=0\label{KKT-O4-1}\\
&\frac{d\mathfrak{L}}{dg}=P(x_k \in {\cal R}^{0}|{\cal H}_0)+\lambda_1\frac{dP'_{M}(g,f,t)}{dg}+\lambda_2\frac{dP'_{F}(g,f,t)}{dg}+\mu_3-\mu_4=0\label{KKT-O4-2}\\
&\lambda_1(P'_{M}(g,f,t)-\alpha)=0,~P'_{M}(g,f,t)\leq\alpha,~\lambda_1 \geq 0,~ \lambda_2(P'_{F}(g,f,t)-\alpha)=0,~P'_{F}(g,f,t)\leq\beta,~\lambda_2\geq0\label{KKT-O4-4}\\
&\mu_1(f-1)=0,~f\leq1,~\mu_1\geq0,~~
\mu_2f=0,~f\geq 0,~\mu_2\geq0,~ \mu_3(g-1)=0,~g\leq1,~\mu_3\geq0,~~
\mu_4g=0,~g\geq 0,~\mu_4\geq0\label{KKT-O4-5}\\
&\frac{d\mathfrak{L}}{df}|_{f=1,g=0} +  k  \frac{d\mathfrak{L}}{dg}|_{f=1,g=0} =P(x_k \in {\cal R}^{-1}|{\cal H}_0)+ \lambda_1\frac{dP'_{M}(g,f,t)}{df}|_{f=1,g=0} +\lambda_2\frac{dP'_{F}(g,f,t)}{df}|_{f=1,g=0}+\mu_1 \nonumber\\
&  + k P(x_k \in {\cal R}^0|{\cal H}_0)  + k \lambda_1\frac{dP'_{M}(g,f,t)}{dg}|_{f=1,g=0}  + k \lambda_2\frac{dP'_{F}(g,f,t)}{df}|_{f=1,g=0}  - k \mu_4 \nonumber\\& =\lambda_1 (1+k) \frac{dP'_{M}(g,f,t)}{df}|_{f=1}+\lambda_2 (1+k) \frac{dP'_{F}(g,f,t)}{df}|_{f=1}+\mu_1+\frac{P(x_k \in {\cal R}^{-1}|{\cal H}_0)}{P(x_k \in {\cal R}^0|{\cal H}_0)}\mu_4\overset{(a)}{>}0 \label{last-equation-appendix-c}
\end{align}
\hrulefill
\end{figure*}
Considering the KKT conditions and the results of Theorem \ref{dP_M_df_thm_opt_i} (which states that given $t$ we have $\frac{dP'_{M}(g(f), f,t)}{df}|_{f=1} \approx 0$ and $
\frac{dP'_{F}(g(f), f,t)}{df}|_{f=1}>0$), next we show\footnote{Note that although we find $f_i^*, g_i^*$ differently in Section \ref{solve-S-random-i}, they still satisfy the KKT conditions in (\ref{KKT-O4-1})-(\ref{KKT-O4-5}) and we use this fact to show that $0 \leq f^{*}_i<1$ and $0<g^{*}_i \leq 1$ under the stated condition in Corollary \ref{cor-random1}.
} that $0 \leq f^{*}_i<1$ and $0<g^{*}_i \leq 1$ under the stated condition in Corollary \ref{cor-random1}. 
For the moment, suppose $f^{*}_i=1, g^{*}_i=0$. From (\ref{KKT-O4-5}), we have $\mu_1> 0$, $\mu_2=0$, $\mu_3=0$, $\mu_4> 0$. 
Also, from the earlier definition of $g(f)$ we have $g(f^{*}_i)=g^{*}_i$, i.e, $g(1)=0$. Furthermore, we have $\frac{dg(f)}{df}=-\frac{P(x_k \in {\cal R}^{-1}|{\cal H}_0)}{P(x_k \in {\cal R}^0|{\cal H}_0)}$, which is fixed (independent of  $g,f$). Let $k=-\frac{P(x_k \in {\cal R}^{-1}|{\cal H}_0)}{P(x_k \in {\cal R}^0|{\cal H}_0)}$ be the fixed ratio where $-1 < k <0$. 
Now, using (\ref{KKT-O4-1}) and (\ref{KKT-O4-2}),  we write (\ref{last-equation-appendix-c})
where $(a)$ follows from  $\lambda_1,\lambda_2,\mu_1,\mu_4, 1+k >0$ and  the results of Theorem \ref{dP_M_df_thm_opt_i}. The inequality in (\ref{last-equation-appendix-c}) suggests that 
we cannot have $\frac{d\mathfrak{L}}{df}|_{f=1,g=0}=0$ and $\frac{d\mathfrak{L}}{dg}|_{f=1,g=0}=0$ simultaneously. Hence $f^{*}_i=1, g^{*}_i=0$ cannot be the solution.  Since $(\mathcal{S}''_1)$ has a solution and $f^{*}_i=1, g^{*}_i=0$ is not a solution, we conclude that $0 \leq f^{*}_i<1$ and $0<g^{*}_i \leq 1$.


\bibliographystyle{IEEETran}
\normalem

\onecolumn
\begin{table}
	\tiny
	\vspace*{4pt}
	\caption{performance comparison when solving $(\mathcal{O})$:  SNR$_h=5$db, SNR$_c=10$db, $\beta=0.01$,  left table $\rho=0.5$ and right table $\rho=0.7$, $p_0$ varies.}
	\hspace{1cm}
	\begin{tabular}{||l|c|c|c|c|l||}
		\hline
		$$&$p_0$&$f^*$&$g^*$&$P_F$&$P_M$\\
		\hline
		pure censoring&$0.4$& $1$& $0$ & $0.01$&$0.1266$\\
		\hline
		CRT-II &$ $&$0.40$& $0.38$ & $0.01$&$0.1036$\\
		\hline
		CRT-I $f=1$ at FC  &$ $&$0.93$& $0.045$ & $0.01$&$0.1260$\\
		\hline
		CRT-I &$ $&$0.47$& $0.33$ & $0.01$&$0.1108$\\
		\hline
		\hline
		$$&$p_0$&$f^*$& $g^*$ & $P_F$&$P_M$\\
		\hline
		pure censoring&$0.6$&$1$& $0$ & $0.01$&$0.1097$\\
		\hline
		CRT-II &$$&$0.5$& $0.72$ & $0.01$&$0.0742$\\
		\hline
		CRT-I $f=1$  at FC &$$&$0.93$& $0.123$ & $0.01$&$0.1050$\\
		\hline
		CRT-I &$$&$0.58$& $0.6$ & $0.01$&$0.0846$\\
		\hline
		\hline
		$$&$p_0$&$f^*$& $g^*$ & $P_F$&$P_M$\\
		\hline
		pure censoring&$0.8$&$1$& $0$ & $0.01$&$0.0824$\\
		\hline
		CRT-II &$$&$0.75$& $1$ & $0.01$&$0.0644$\\
		\hline
		CRT-I $f=1$ at FC  &$ $&$1$ & $0$ &$0.01$&$0.0824$\\
		\hline
		CRT-I &$ $&$0.89$ & $0.39$ &$0.01$&$0.0800$\\
		\hline
		\end{tabular}
		\quad
		\begin{tabular}{||l|c|c|c|c|l||}
		\hline
		$$&$p_0$&$f^*$& $g^*$ & $P_F$&$P_M$\\
		\hline
		pure censoring&$0.4$&$1$& $0$ & $0.01$&$0.1500$\\
		\hline
		CRT-II &$ $&$0.38$& $0.39$ & $0.01$&$0.1250$\\
		\hline
		CRT-I $f=1$ at FC  &$ $&$0.93$& $0.049$ & $0.01$&$0.1390$\\
		\hline
		CRT-I &$ $&$0.47$& $0.32$ & $0.01$&$0.1270$\\
		\hline
		\hline
		$$&$p_0$&$f^*$& $g^*$ & $P_F$&$P_M$\\
		\hline
		pure censoring&$0.6$&$1$& $0$ & $0.01$&$0.1424$\\
		\hline
		CRT-II &$ $&$0.47$& $0.76$ & $0.01$&$0.1144$\\
		\hline
		CRT-I $f=1$ at FC  &$ $&$0.93$& $0.108$ & $0.01$&$0.1324$\\
		\hline
		CRT-I &$ $&$0.47$& $0.76$ & $0.01$&$0.1189$\\
		\hline
		\hline
		$$&$p_0$&$f^*$& $g^*$ & $P_F$&$P_M$\\
		\hline
		pure censoring&$0.8$&$1$& $0$ & $0.01$&$0.1132$\\
		\hline
		CRT-II &$$&$0.77$& $0.905$ & $0.01$&$0.0962$\\
		\hline
		CRT-I $f=1$ at FC  &$$&$0.93$& $0.264$ & $0.01$&$0.1100$\\
		\hline
		CRT-I &$$&$0.79$& $0.78$ & $0.01$&$0.1040$\\
		\hline
		\end{tabular}
		\label{table_problem_O_rho_07_w/channel}
		\end{table}

		
		
		\begin{table}
		\tiny
		\caption{performance comparison when solving $(\mathcal{O})$:  SNR$_h=10$db, SNR$_c=10$db, $\beta=0.01$, $\rho=0.5$, $p_0$ varies.}
		\hspace{1cm}
		\begin{tabular}{||l|c|c|c|c|l||}
		\hline
		$$&$p_0$&$f^*$&$g^*$& $P_F $&$P_M $\\
		\hline
		pure censoring&$0.4$&$1$& $ 0$ & $0.01$&$0.0593$\\
		\hline
		CRT-II &$ $&$0.45$& $ 0.35$ & $0.01$&$0.0530$\\
		\hline
		CRT-I $f=1$  at FC &$ $&$0.93$& $ 0.045$ & $0.01$&$0.0580$\\
		\hline
		CRT-I &$ $&$0.52$& $ 0.43$ & $0.01$&$0.0540$\\
		\hline
		\hline
		$$&$p_0$&$f^*$&$g^*$& $P_F $&$P_M $\\
		\hline
		pure censoring&$0.6$&$1$ & $ 0$  &$0.01$&$0.0548$\\
		\hline
		CRT-II &$ $&$0.87$& $ 0.20$ & $0.01$&$0.0490$\\
		\hline
		CRT-I $f=1$ at FC  &$ $&$0.93$& $ 0.085$ & $0.01$&$0.0530$\\
		\hline
		CRT-I &$ $&$0.87$& $ 0.20$ & $0.01$&$0.0498$\\
		\hline
		\hline
		$$&$p_0$&$f^*$& $g^*$& $P_F $&$P_M $\\
		\hline
		pure censoring&$0.8$&$1$& $ 0$ & $0.01$&$0.0426$\\
		\hline
		CRT-II &$ $&$0.87$& $ 0.51$ & $0.01$&$0.0420$\\
		\hline
		CRT-I  $f=1$  at FC &$ $&$1$& $ 0$ & $0.01$&$0.0426$\\
		\hline
		CRT-I &$ $&$0.92$& $ 0.32$ & $0.01$&$0.0422$\\
		\hline
	\end{tabular}
	\label{table_problem_O_rho_05_wo/channel}
\end{table}

\begin{table}
	\tiny
	\caption{performance comparison  when solving $(\mathcal{O})$. Left Table: SNR$_h=5$db, SNR$_c=10$db, $\beta=0.01$, $p_0=0.4$, $\rho$ varies. Right Table: SNR$_h=5$db, SNR$_c=12$db, $\beta=0.01$, special case of $\rho=0$, $p_0$ varies.}
	\hspace{1cm}
	\begin{tabular}{||l|c|c|c|c|l||}
		\hline
		$$&$\rho$&$f^*$& $g^*$ & $P_F$&$P_M$\\
		\hline
		pure censoring&$0.1$&$1$& $0$ & $0.01$&$0.0595$\\
		\hline
		CRT-II &$ $&$1$& $0$ & $0.01$&$0.0595$\\
		\hline
		CRT-I $f=1$ at FC &$ $&$1$& $0$ & $0.01$&$0.0595$\\
		\hline
		CRT-I &$ $&$1$&$0$ & $0.01$&$0.0595$\\
		\hline
		\hline
		$$&$\rho$&$f^*$& $g^*$ & $P_F$&$P_M$\\
		\hline
		pure censoring&$0.3$&$1$& $0$ & $0.01$&$0.0648$\\
		\hline
		CRT-II &$ $&$0.54$& $0.29$ & $0.01$&$0.0600$\\
		\hline
		CRT-I $f=1$  at FC &$ $&$1$& $0$ & $0.01$&$0.0648$\\
		\hline
		CRT-I &$ $&$0.58$& $0.26$ & $0.01$&$0.0610$\\
		\hline
		\hline
		$$&$\rho$&$f^*$& $g^*$ & $P_F$&$P_M$\\
		\hline
		pure censoring&$0.5$&$1$& $0$ & $0.01$&$0.1266$\\
		\hline
		CRT-II &$ $&$0.4$& $0.39$ & $0.01$&$0.1036$\\
		\hline
		CRT-I $f=1$ at FC  &$ $&$0.93$& $0.045$ & $0.01$&$0.1260$\\
		\hline
		CRT-I &$ $&$0.47$& $0.33$ & $0.01$&$0.1108$\\
		\hline
		\hline
		$$&$\rho$&$f^*$& $g^*$ & $P_F$&$P_M$\\
		\hline
		pure censoring&$0.7$&$1$& $0$ & $0.01$&$0.1500$\\
		\hline
		CRT-II &$ $&$0.38$& $0.39$ & $0.01$&$0.1250$\\
		\hline
		CRT-I $f=1$ at FC  &$ $&$0.93$& $0.049$ & $0.01$&$0.1390$\\
		\hline
		CRT-I &$ $&$0.47$& $0.32$ & $0.01$&$0.1270$\\
		\hline
		\hline
		$$&$\rho$&$f^*$& $g^*$ & $P_F$&$P_M$\\
		\hline
		pure censoring&$0.9$&$1$& $0$ & $0.01$&$0.1800$\\
		\hline
		CRT-II &$ $&$0.36$& $0.4$ & $0.01$&$0.1566$\\
		\hline
		CRT-I $f=1$ at FC &$ $&$0.8$& $0.148$ & $0.01$&$0.1770$\\
		\hline
		CRT-I &$ $&$0.45$& $0.34$ & $0.01$&$0.1700$\\
		\hline
		\end{tabular}
		\quad
		\begin{tabular}{||l|c|c|c|c|l||}
		\hline
		$$&$p_0$&$f^*$& $g^*$ & $P_F$&$P_M$\\
		\hline
		pure censoring&$0.4$&$1$& $0$ & $0.01$&$0.0109$\\
		\hline
		CRT-II &$$&$0.56$& $0.28$ & $0.01$&$0.0099$\\
		\hline
		CRT-I &$$&$1$& $0$ & $0.01$&$0.0109$\\
		\hline
		\hline
		$$&$p_0$&$f^*$& $g^*$ & $P_F$&$P_M$\\
		\hline
		pure censoring&$0.5$&$1$& $0$ & $0.01$&$0.0052$\\
		\hline
		CRT-II &$$&$0.89$& $0.30$ & $0.01$&$0.0048$\\
		\hline
		CRT-I &$$&$1$& $0$ & $0.01$&$0.0052$\\
		\hline
		\hline
		$$&$p_0$&$f^*$& $g^*$ & $P_F$&$P_M$\\
		\hline
		pure censoring&$0.8$&$1$& $0$ & $0.01$&$0.000585$\\
		\hline
		CRT-II &$$&$1$& $0$ & $0.01$&$0.000585$\\
		\hline
		CRT-I &$$&$1$& $0$ & $0.01$&$0.000585$\\
		\hline
	\end{tabular}
	\label{table_problem_O_rho_varies_w/channel}
\end{table}
\begin{table}
	\tiny
	\vspace*{4pt}
	\caption{Effect of correlation mismatch on pure censoring performance when solving $(\mathcal{O})$: SNR$_h=5$db, SNR$_c=10$db, $\beta=0.01$, $\rho=0.5$, $p_0$ varies.}
	\hspace{-1.5cm}
	\begin{tabular}{||l|c|c|c|c|c|c|c|c|c|c|c|l||}
		\hline
		$p_0$&$P_F|_{\rho=0}$&$P_M|_{\rho=0}$&$P_F|_{\rho=0.1}$&$P_M|_{\rho=0.1}$&$P_F|_{\rho=0.3}$&$P_M|_{\rho=0.3}$&$P_F|_{\rho=0.5}$&$P_M|_{\rho=0.5}$&$P_F|_{\rho=0.7}$&$P_M|_{\rho=0.7}$&$P_F|_{\rho=0.9}$&$P_M|_{\rho=0.9}$\\
		\hline
		$0.4$&$0.01$&$0.0315$&$0.0148$&$0.0527$&$0.151$&$0.0652$&$0.0192$&$0.0695$&$0.0232$&$0.0814$&$0.0262$&$0.0906$\\
		\hline
		$0.6$&$0.01$&$0.0172$&$0.0119$&$0.0193$&$0.0188$&$0.0318$&$0.0276$&$0.0476$&$0.0308$&$0.0557$&$0.0372$&$0.0697$\\
		\hline
		$0.8$&$0.01$&$0.0048$&$0.0251$&$0.0137$&$0.0465$&$0.0198$&$0.0688$&$0.0254$&$0.1018$&$0.0257$&$0.1539$&$0.0217$\\
		\hline
	\end{tabular}
	\label{table_correlation_mismatch}
\end{table}


\begin{table}
	\tiny
	\caption{performance comparison when solving $(\mathcal{S})$.  Left Table: SNR$_h=5$db, SNR$_c=10$db, $\beta=0.01$, $p_0=0.4$, $\rho$ varies.  Right Table: SNR$_h=5$db, SNR$_c=12$db, $\beta=0.01$, special case of $\rho=0$, $p_0$ varies.}
	\hspace{-1cm}
	\begin{tabular}{|l|c|c|c|c|c|c|l||}
		\hline
		\hline
		$$&$\rho$&$t$&$f^*$&$g^*$&$P_F$&$P_M$&$P_t$\\
		\hline
		pure censoring&$0.1$&$7$&$1$&$0$&$0.01$&$0.1$&$0.0196$\\
		\hline
		CRT-II&$$&$7$&$1$&$0$&$0.01$&$0.1$&$0.0196$\\
		\hline
		CRT-I at FC $f=1$ $(\mathcal{S}''_1)$&$$&$7$&$1$&$0$&$0.01$&$0.1$&$0.0196$\\
		\hline
		CRT-I&$$&$7$&$1$&$0$&$0.01$&$0.1$&$0.0196$\\
		\hline
		\hline
		$$&$\rho$&$t$&$f^*$&$g^*$&$P_F$&$P_M$&$P_t$\\
		\hline
		pure censoring&$0.3$&$7$&$1$&$0$&$0.01$&$0.1$&$0.0565$\\
		\hline
		CRT-II&$$&$7$&$0.8$&$0.0018$&$0.01$&$0.1$&$0.0534$\\
		\hline
		CRT-I at FC $f=1$ &$$&$6.5$&$0.98$&$10^{-4}$&$0.01$&$0.1$&$0.0550$\\
		\hline
		CRT-I&$$&$6.7$&$0.8$&$0.0033$&$0.01$&$0.1$&$0.0546$\\
		\hline
		\hline
		$$&$\rho$&$t$&$f^*$&$g^*$&$P_F$&$P_M$&$P_t$\\
		\hline
		pure censoring&$0.5$&$6$&$1$&$0$&$0.01$&$0.1$&$0.3328$\\
		\hline
		CRT-II&$$&$5.5$&$0.62$&$0.06$&$0.01$&$0.1$&$0.2533$\\
		\hline
		CRT-I $f=1$ at FC &$$&$5.5$&$0.94$&$0.01$&$0.01$&$0.1$&$0.3100$\\
		\hline
		CRT-I&$$&$5.7$&$0.76$&$0.050$&$0.01$&$0.1$&$0.2915$\\
		\hline
		\hline
		$$&$\rho$&$t$&$f^*$&$g^*$&$P_F$&$P_M$&$P_t$\\
		\hline
		pure censoring&$0.7$&$6$&$1$&$0$&$0.01$&$0.1$&$0.8193$\\
		\hline
		CRT-II&$$&$5.0$&$0.6$&$0.53$&$0.01$&$0.1$&$0.5890$\\
		\hline
		CRT-I $f=1$ at FC&$$&$6.0$&$0.95$&$10^{-3}$&$0.01$&$0.1$&$0.77$\\
		\hline
		CRT-I&$$&$4.0$&$0.68$&$0.25$&$0.01$&$0.1$&$0.61$\\
		\hline
		\hline
		$$&$\rho$&$t$&$f$&$g$&$P_F$&$P_M$&$P_t$\\
		\hline
		pure censoring&$0.9$&$6$&$1$&$0$&$0.01$&$0.1$&$0.8676$\\
		\hline
		CRT-II&$$&$5.5$&$0.6$&$0.8$&$0.01$&$0.1$&$0.6292$\\
		\hline
		CRT-I $f=1$ at FC&$$&$6$&$0.92$&$0.01$&$0.01$&$0.1$&$0.81$\\
		\hline
		CRT-I&$$&$4.0$&$0.68$&$0.34$&$0.01$&$0.1$&$0.65$\\
		\hline	
	\end{tabular}
	\quad
	\begin{tabular}{|l|c|c|c|c|c|l||}
		\hline
		$$&$t$&$f^*$&$g^*$&$P_F$&$P_M$&$P_t$\\
		\hline
		pure censoring&$2.2$&$1$&$0$&$0.01$&$0.025$&$0.2427$\\
		\hline
		CRT-II&$2.0$&$0.32$&$0.17$&$0.01$&$0.025$&$0.2106$\\
		\hline
		CRT-I &$2.2$&$1$&$0$&$0.01$&$0.025$&$0.2427$\\
		\hline
		\end{tabular}
		\label{table_problem_S_rho_varies_w/channel}
		\end{table}


\newpage

\begin{figure}[h]
\centering
\subfigure[]{
\includegraphics[scale=0.3]{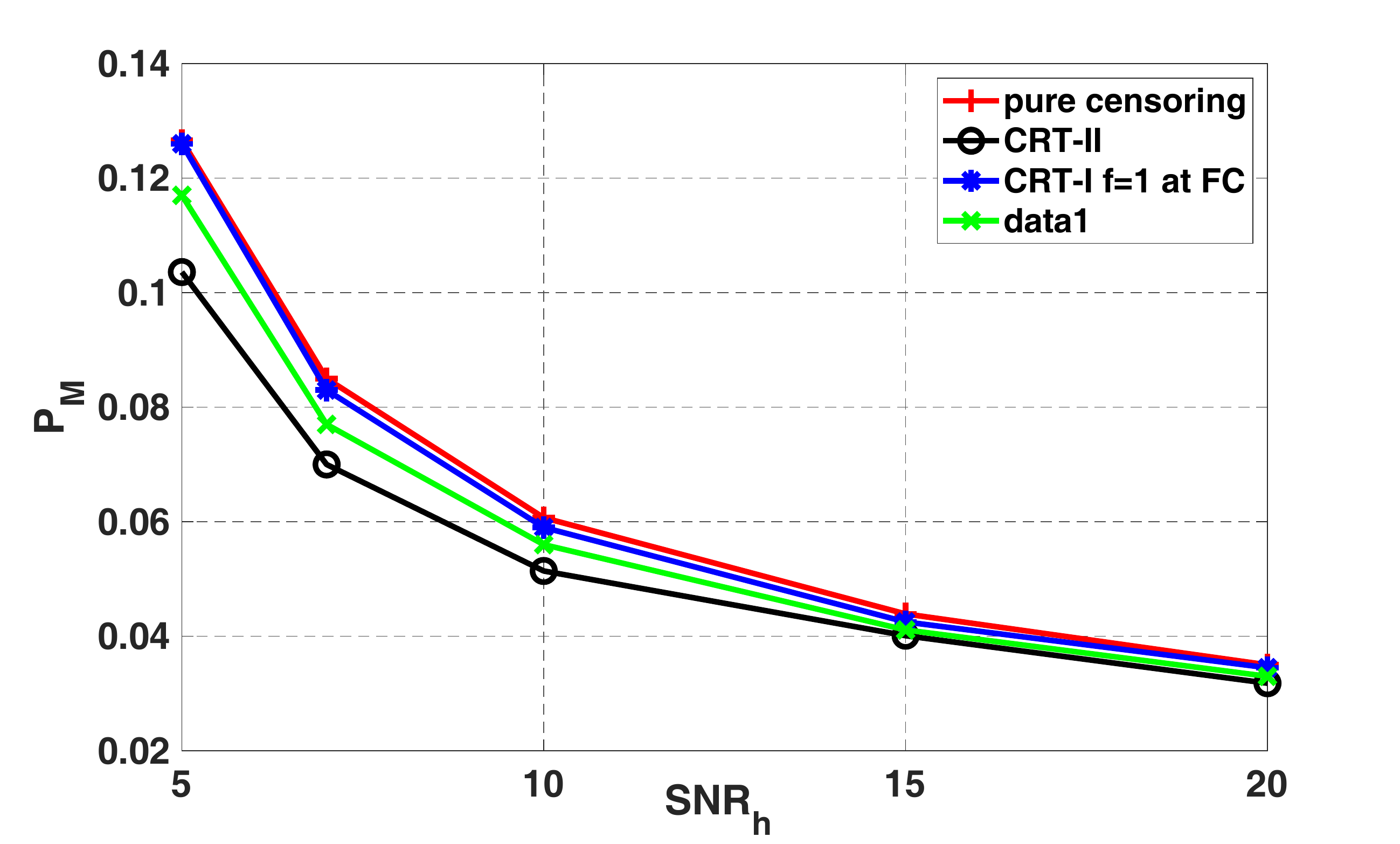}
\label{P_M-versus-channel-fig}
}
\subfigure[]{
\includegraphics[scale=0.3]{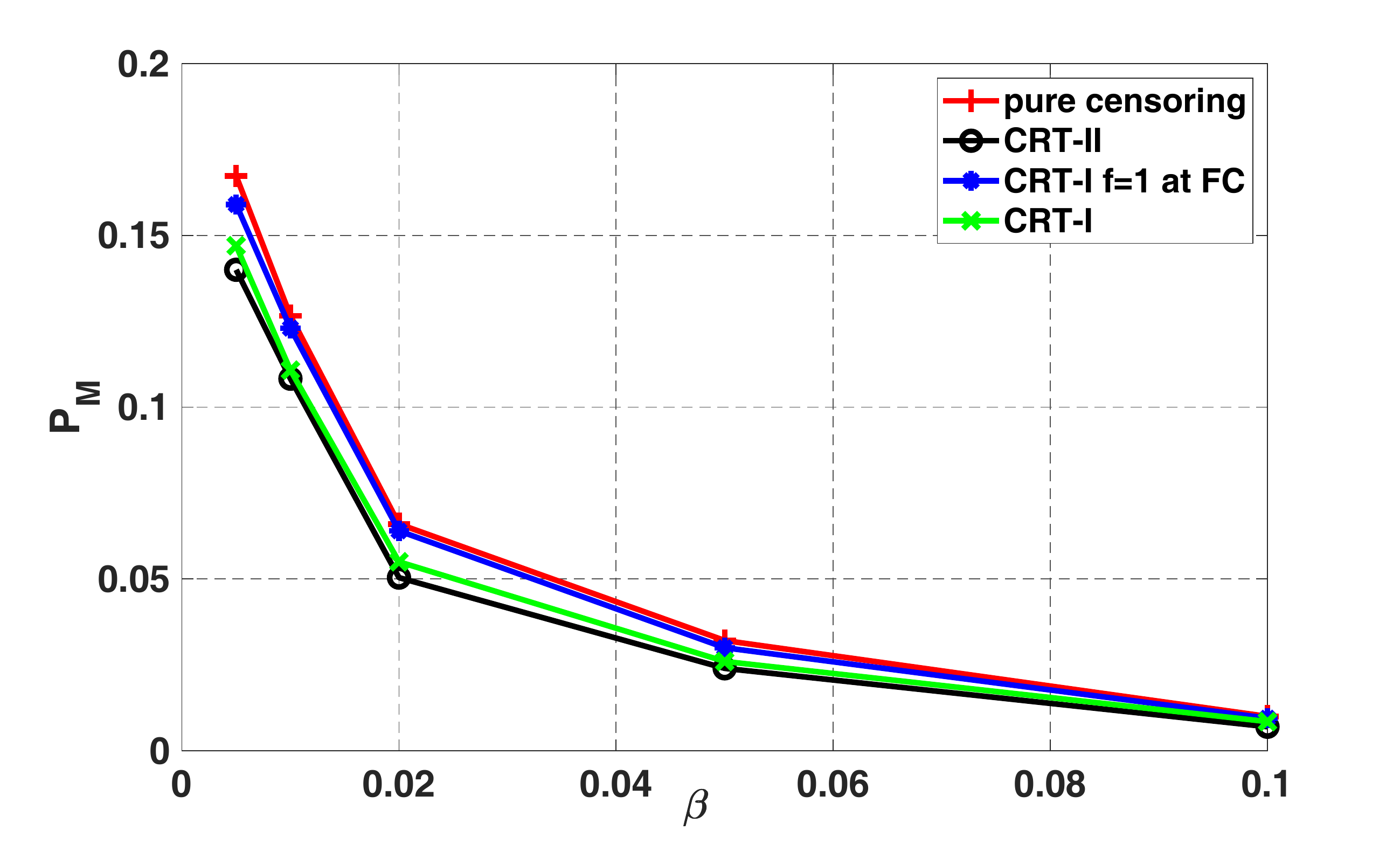}
\label{P_M-versus-P_f-fig}
}
\caption{   Solving $(\mathcal{O})$, SNR$_c=10$dB, $p_0=0.4$, $\rho=0.5$ (a) $P_M$ versus SNR$_h$ for   $\beta=0.01$. (b) $P_M$ versus $\beta$ for SNR$_h=10$dB.}
\label{}
\end{figure}

\begin{figure}[h!]
\centering
\subfigure[]{
	\includegraphics[scale=0.3]{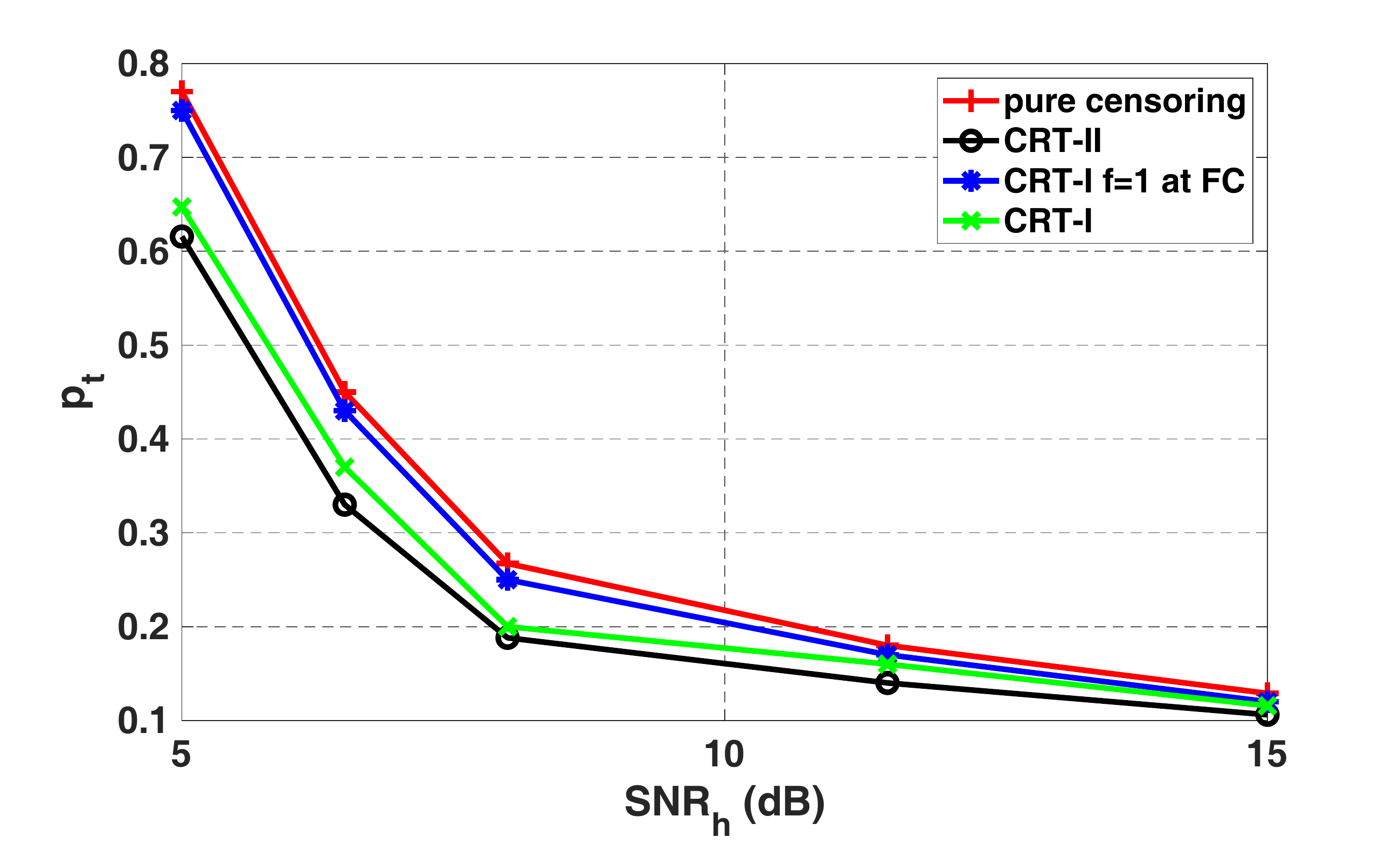}
	\label{p_t-versus-channel-fig}
}
\subfigure[]{
	\includegraphics[scale=0.3]{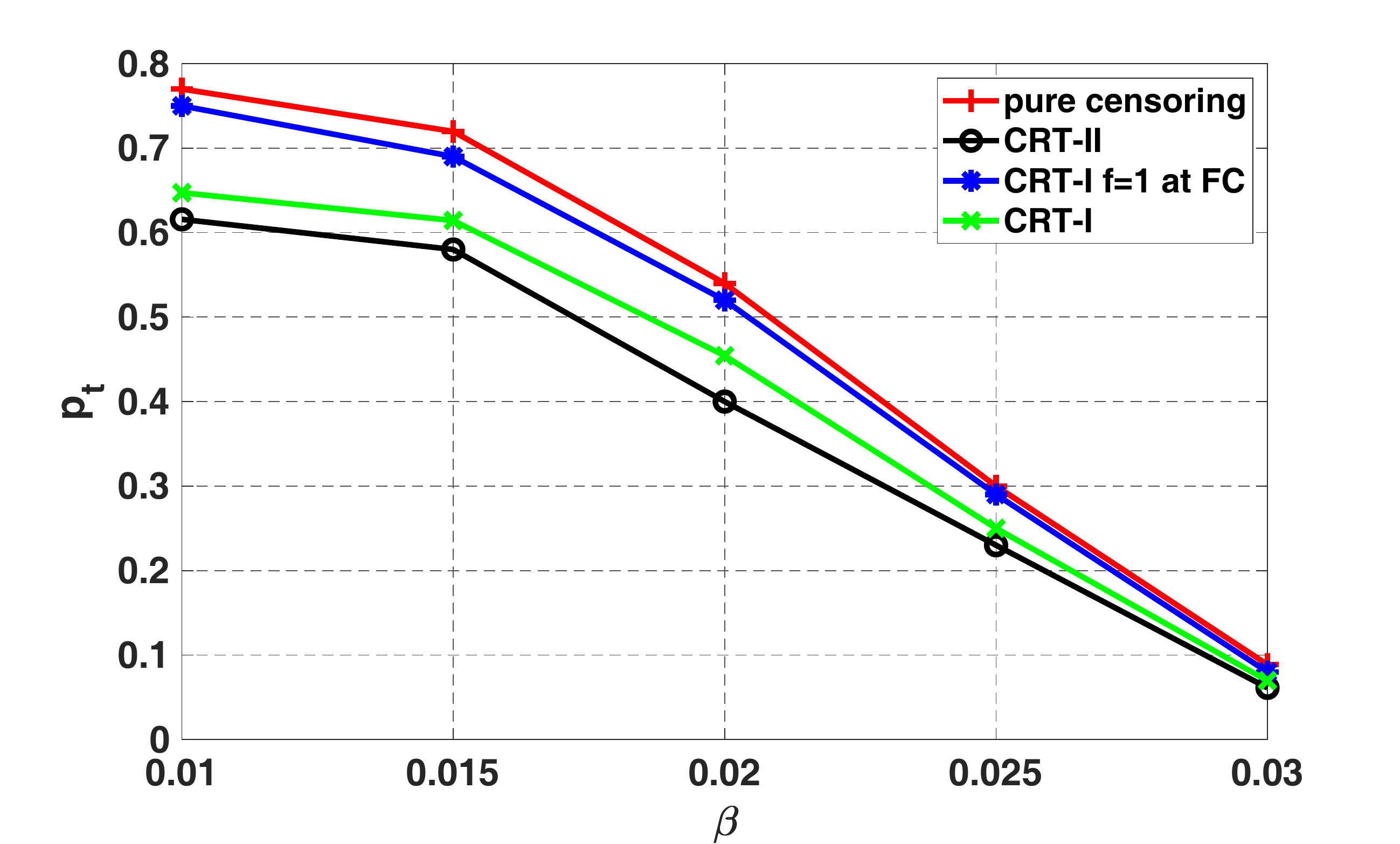}
	\label{p_t-versus-P_f-fig}
}
\caption{   Solving $(\mathcal{S})$, SNR$_c=10$dB, $\alpha=0.06$, $\rho=0.5$ (a) $P_t$ versus SNR$_h$ for $\beta=0.01$. (b) $P_t$ versus $\beta$ for SNR$_h=10$dB.}
\label{}
\end{figure}

\end{document}